\documentclass[12pt]{article}
\usepackage{url}
\usepackage{authblk}

\newcommand{\comment}[1]{}

\RequirePackage[OT1]{fontenc}
\RequirePackage{amsthm,amsmath}

\usepackage{mathrsfs}
\usepackage{amsmath, amssymb, amsthm, amsbsy, amsfonts, cases, color}

\usepackage{mathtools,dsfont}
\usepackage{hyperref}
\usepackage{subfigure}
\usepackage{graphics}
\usepackage{epsfig}
\usepackage{epstopdf}

\newcommand{\X}{\textbf{X}}
\newcommand{\Y}{\textbf{Y}}
\newcommand{\Xt}{\textbf{X}_t}
\newcommand{\Yt}{\textbf{Y}_t}
\newcommand{\B}{\textbf{B}}
\newcommand{\Bt}{\textbf{B}_t}
\newcommand{\BtTilde}{\tilde{\textbf{B}_t}}

\newcommand{\alphabold}{\boldsymbol{\alpha}}
\newcommand{\alphau}{\boldsymbol{\alpha}_u}
\newcommand{\alphat}{\boldsymbol{\alpha}_t}

\newcommand{\alphab}{\boldsymbol{\alpha}}
\newcommand{\alphabEst}{\hat{\boldsymbol{\alpha}}}
\newcommand{\betamat}{\boldsymbol{\beta}}
\newcommand{\Xvec}{\boldsymbol{x}}

\newcommand{\Amatrix}{\boldsymbol{A}}
\newcommand{\AmatrixEst}{\hat{\boldsymbol{A}}}

\newcommand{\Dmatrix}{\boldsymbol{D}}

\newcommand{\Gammamat}{\boldsymbol{\Gamma}}
\newcommand{\Fmatrix}{\boldsymbol{F}}

\newcommand{\lambdamat}{\boldsymbol{\lambda}}
\newcommand{\etamat}{\boldsymbol{\eta}}
\newcommand{\Psimat}{\boldsymbol{\Psi}}
\newcommand{\Imatrix}{\boldsymbol{I}}
\newcommand{\Pmatrix}{\boldsymbol{P}}
\newcommand{\Qmatrix}{\boldsymbol{Q}}

\newcommand{\Lmatrix}{\boldsymbol{L}}
\newcommand{\kappamat}{\boldsymbol{\kappa}}

\newcommand{\kappamatEst}{\hat{\boldsymbol{\kappa}}}

\newcommand{\varphimat}{\boldsymbol{\varphi}}
\newcommand{\sigmamat}{\boldsymbol{\sigma}}
\newcommand{\sigmamatEst}{\hat{\boldsymbol{\sigma}}}
\newcommand{\thetavec}{\boldsymbol{\theta}}
\newcommand{\rhomatrix}{\boldsymbol{\Theta}}

\newcommand{\Corrmatrix}{\rhomatrix}
\newcommand{\CorrmatrixEst}{\hat{\rhomatrix}}
\newcommand{\limrhoI}{\lim_{\rhomatrix \rightarrow \Imatrix}}
\newcommand{\Indicator}{\mathcal{I}}

\newcommand{\corridx}{{mn}}
\newcommand{\Is}{\Imatrix^\corridx}
\newcommand{\corridxsecond}{{pq}}
\newcommand{\Ik}{\Imatrix^\corridxsecond}
\newcommand{\Ixi}{\Imatrix^{uu}}
\newcommand{\rhos}{\rho_\corridx}
\newcommand{\rhok}{\rho_\corridxsecond}
\newcommand{\E}{\mathbb{E}}
\newcommand{\R}{\mathfrak{R}}
\newcommand{\Qeq}{\mathfrak{B}}

\newcommand{\partialX}{\boldsymbol{\nabla}}
\newcommand{\sums}{\sum_{s=1}^n}
\newcommand{\sumk}{\sum_{k=1}^n}
\newcommand{\exparg}{2\kappa \sqrt{\delta} \tau}

\newcommand{\expargi}{2\kappa_i \sqrt{\delta} \tau}
\newcommand{\expargj}{2\kappa_j \sqrt{\delta} \tau}

\newcommand{\argmax}{\operatornamewithlimits{argmax}}

\newlength{\querylen}
\usepackage{fancybox}

\theoremstyle{plain}
\newtheorem{theorem}{Theorem}[section]

\newtheorem{lemma}[theorem]{Lemma}

\begin{document}

\date{\today}

\title{Trading multiple mean reversion}
\author[1]{E. Boguslavskaya\thanks{$^\ast$
		Email: elena.boguslavskaya@brunel.ac.uk}}
\author[2]{M. Boguslavsky\thanks{$^\ast$
	Email: michael@boguslavsky.net}}
\author[3]{D. Muravey\thanks{$^\ast$
   Email: d.muravey87@gmail.com}}
\affil[1]{Brunel University London, UK}
\affil[2]{TradeTeq, UK}
\affil[3]{Lomonosov State University, Moscow, Russia}

\maketitle
\begin{abstract}{
How should one construct a portfolio from multiple mean-reverting assets?
Should one add an asset to portfolio even if the asset has zero mean reversion?
We consider a position management problem for an agent trading multiple
mean-reverting assets. We solve an optimal control problem for an agent
with power utility, and present a semi-explicit solution. The nearly explicit
nature of the solution allows us to study the effects of parameter
mis-specification, and derive a number of properties of the optimal solution.}
\end{abstract}

%

\tableofcontents
\section{Introduction}
One of the basic patterns of statistical arbitrage is mean reversion trading. Typically, one constructs a synthetic asset from one or several traded assets in such a way that its price dynamics is mean reverting. For example, for a pair of cointegrated assets there exists a mean-reverting linear combination of these assets. We will be calling this mean-reverting synthetic asset 
the {\it spread}. Generally, trading a mean reverting asset consists of buying the spread when it is below its mean level and sellings when it is above. The main question is how should the position be optimally managed with movement of the spread, trader's risk aversion, and time horizon. When there are several mean-reverting assets available, the trader should additionally solve a dynamic portfolio optimization problem in order to decide the best way to combine positions in these assets.

A number of papers addressed this problem by specifying a stochastic differential equation (SDE) for spread dynamics and finding the optimal strategy that optimizes the expected utility over the terminal wealth. The simplest example of mean-reverting dynamics in continuous time is the Ornstein--Uhlenbeck process, the continuous version of the AR(1) discrete process. For a single spread optimal trading strategy see \cite{BG}. For a more complicated mean-reverting dynamics we refer to paper \cite{ALTAY}, where the spread is modelled by a Markov modulated Ornstein--Uhlenbeck process, and to papers \cite{Fouque1} and \cite{Fouque2} where the authors consider fractional stochastic processes. The models with uncertainty in the mean reversion level were discussed in \cite{PAP01}. Other models for the spread have also been considered in the literature: for models based on Brownian brigde see \cite{LILO}, and for models based on CER/CIR processes see \cite{ZER}. A comprehensive review of the mean reversion trading can be found in \cite{LEUNG}. For methodology of statistical arbitrage we refer to \cite{AVELANEDA}. In \cite{PAP02} the authors assume different mean-reversion dynamics for multiple spread processes. They solve a portfolio optimization problem for several Geometric Brownian motions with multiple co-integration terms in drifts.

Usually a portfolio allocator has access to multiple investing opportunities.  Optimal sizing and timing of positions in each of these opportunities may be affected by positions in other assets and performance of those assets. To develop intuition about optimal dynamic allocation strategy, we generalise \cite{BG} to the case of multiple correlated Ornstein-Uhlenbeck and Brownian Motion processes. We solve the problem of maximization of a power utility over the terminal wealth for a finite horizon agent. Power utilities are a sufficiently broad family of utility functions, containing log-utility as a special case and linear utility as a limit case.

For the general problem, the optimal strategy is found in quasi--analytical form as a solution to a matrix Riccati ordinary differential equation. For several important special cases it is possible to solve this equation explicitly. We also propose an efficient approach to analyse effects of parameter mis-specification. Although the proposed model is very simple, one can observe non-trivial qualitative properties of the optimal strategy. The availability of a quasi--analytical solution allows us to study how the trading strategy is affected by correlation between spreads, and demonstrate the tradeoffs between "harvesting" each spread separately and hedging positions in correlated spreads.

The rest of this paper is organized as follows: in Section \ref{sec:Results} we give a brief overview of optimal strategy properties. In Section \ref{sec:Model} we specify our formal asset and trading model  and formulate a stochastic optimal control problem. Section \ref{sec:MainResult} contains explicit formulas for the optimal control and the value function. Section \ref{sec:Analysis1D} reminds main insights for the one-dimensional case. Optimal solution analysis is presented in Section \ref{sec:AnalysisMD}. In Section \ref{sec:MISSPEC}, we present an ODE based framework to analyse the effect of parameter mis-specification and calculate the moments of the terminal wealth's distribution. We then apply this framework to analyse strategy and value sensitivity to reversion rates misspecification.

Implementation source code in python and numerical implementation hints are available at \cite{CODE}.

\section{Main results}
\label{sec:Results}
The optimal solution has a number of interesting qualitative properties.

\begin{itemize}
	\item {\it Trade-off between hedging and spread extraction}
	\\
	In the case of a single asset, the position is managed to extract value from this asset movements. With several correlated mean-reverting assets, the optimal strategy also uses positions in assets with slower mean reversion to hedge positions in faster mean reverting assets.
	
	%
	\item {\it Impact of correlations}
	\\
	With all other parameters fixed, higher absolute values of correlations between asset driving processes are preferable to lower absolute values, as long as they stay below 1. See Section \ref{sec:correlation} for more details.
	
	\item {\it Impact of different reversion rates}
	\\
	With all other parameters fixed, higher reversion speeds are not always preferable for the trader. An asset with a lower reversion rate and a non-zero correlation with higher reversion rate assets, may be used primarily as a hedge for positions in these assets. Hedge efficiency may be declining with the increases in the lower reversion rate. See Section \ref{sec:example2d} for more details.
	
	\item {\it Cost of parameter misspecification}
	\\
	The optimal strategy has a strong dependence on assumed reversion rates. It is safer to underestimate reversion rates than to  overestimate them. The value function is more sensitive to errors in reversion rate ratios between assets than to joint correlated errors in rate estimates. See Section \ref{sec:MISSPEC}.
	
\end{itemize}

\section{The model}
\label{sec:Model}
\subsection{Price processes}
\label{sec:price_process}
Assume the canonical multivariate filtered probability space $\left( \Omega, \, \mathcal{F}, \, \mathbb{F}, \, \mathbb{P} \right)$ with filtration $\left(\mathcal{F}_t \right)_{t \geq 0}$  to satisfy the usual conditions, see e.g. \cite{SHREEVE}. On this space let $\left[X_t^1, X_t^2, \hdots, X_t^n \right]^{\top}$  be a collection of tradeable assets following  a multidimensional Ornstein--Uhlenbeck process
\begin{equation}
\label{eq:OU_def}
d\Xt = -\kappamat \Xt dt + \sigmamat d\Bt
\end{equation}
Here $\Bt = \left[B_t^1, B_t^2, \hdots, B_t^n \right]^{\top} $ is an $n$-dimensional Wiener process with correlation matrix $\rhomatrix \in \mathbb{R}^{n \times n}$ (i.e. $d\Bt d\Bt^{\top} = \Corrmatrix dt$), and
$\kappamat \in \mathbb{R}_{+}^{n \times n}$ and  $\sigmamat \in \mathbb{R}_{+}^{n \times n}$
are diagonal matrices with reversion rates and volatility entries correspondingly
\begin{eqnarray}
\nonumber
\begin{array}{c}
\kappamat = diag(\kappa_1, \kappa_2, \hdots, \kappa_n),
\\
\sigmamat = diag(\sigma_1, \sigma_2, \hdots, \sigma_n),
\end{array}
\quad
\rhomatrix = \begin{bmatrix}
1 & \rho_{12} & \hdots & \rho_{1n} \\
\rho_{21} & 1 & \hdots & \rho_{2n} \\
\vdots & \vdots & \ddots & \vdots \\
\rho_{n1} & \rho_{n2} & \hdots & 1
\end{bmatrix}
\end{eqnarray}
The diagonality of matrices $\kappamat$ and $\sigmamat$ means that all dependency between assets comes from the correlations between the driving Brownian motions. We also consider models with some assets exhibiting zero mean reversion (i.e. with some zero elements of $\kappamat$.) These assets are simply following correlated Brownian motions. However, we assume that elements of vector $\kappamat$ are not all zero to avoid a trivial problem. Correlation matrix $\rhomatrix$ should be symmetric and positive semi-definite with unit diagonal elements, $\rho_{ii} = 1$, $\rho_{ij} = \rho_{ji}$. We will assume that $\rhomatrix$ has full rank to avoid obvious arbitrages.

Without loss of generality, we can also assume that long-term means of each process are equal to zero. The general case can be reduced to equation (\ref{eq:OU_def}) by the substitution $\left[ \Xt - \thetavec \right] \rightarrow \Xt$, where $\thetavec$ is a vector of long term means.
Equation (\ref{eq:OU_def}) can be solved explicitly in terms of It\^o integral:
\begin{equation}
\nonumber
\Xt = e^{-\kappamat t} \X_0  + \int_{0}^t e^{-\kappamat(t - s)} \sigmamat d\B_s
\end{equation}
Here $e^{\Amatrix}$ is a matrix exponential:
\begin{equation}
\nonumber
e^{\Amatrix} = \sum_{k = 0} ^{\infty} \frac{1}{k!} \Amatrix^k, \quad \Amatrix^0 = \Imatrix.
\end{equation}

\subsection{Wealth process}
The problem can be treated in the general Merton portfolio optimisation framework, see \cite{ME}. Let vector $\alphat$
\begin{equation}
\nonumber
\alphat = \left[ \alpha_t^1, \alpha_t^2, \hdots, \alpha_t^n \right]^{\top}
\end{equation}
be a trader’s position at time $t$, i.e. the number of units of each asset held. This is the control in our optimization problem. Assuming zero interest rates and no transaction costs, for a given control process $\alphab_t$, the wealth process
$W_t^{\alphabold}$  is given by
\begin{eqnarray}
\nonumber
\label{eq:W_def}
dW_t^{\alphabold} =\alphat^{\top} d\Xt = \sum_{i=1}^{n} \alpha_t^i dX_t^i,
\end{eqnarray}
or in integral form
\begin{eqnarray}
\nonumber
W_t^{\alphabold} = W_t^{\alphabold} + \int_t^T \alphau^{\top} d \boldsymbol{X}_u = W_t^{\alphabold} + \sum_{i=1}^{n} \int_t^T \alpha_u^i dX_u^i.
\end{eqnarray}
\subsection{Normalization}
Without loss of generality, we assume unit noise magnitudes: i.e. $\sigmamat = \Imatrix$. For the general case, the following parametrisation should be used:
\begin{eqnarray}
\nonumber
\Xt \rightarrow \sigmamat^{-1}\Xt,
\quad \alphab_t \rightarrow \sigmamat \alphab_t.
\end{eqnarray}

\subsection{Value function}
The value function $J(W^{\alphabold}_t, \Xt, t) : \mathbb{R}^{+} \times \mathbb{R}^n \times [0, T] \rightarrow \mathbb{R} $  is the supremum over all admissible controls of the expectation of the terminal utility conditional on the information available at time $t$
\begin{equation}
\nonumber
\label{eq:J_def}
	J(w, \Xvec, t) = \sup_{\alphat \in \mathcal{A}} \E\left[U(W^{\alphabold}_T) | W^{\alphabold}_t = w, \, \Xt = \Xvec  \right],
\end{equation}
where the set of admissible controls $\mathcal{A}$ is defined as
\begin{eqnarray}
\label{eq:Admissible_Set}
\mathcal{A} =  \left\{ \alphab : [0, T] \times \Omega \rightarrow \mathbb{R}^n
\, | \, \alphat \in \mathcal{F}_t ,\,  \int_0^{\top} \left(W^{\alphabold}_t\right)^2 \sum_{i = 1}^n \left(  \alphat^i \Xt^i \right)^2 dt < \infty, \quad a.s \right\}
\end{eqnarray}
We consider a power utility function with the parameter $\gamma < 1$
\begin{equation}
\nonumber
U = U(W^{\alphabold}_T) = \frac{1}{\gamma} \left(W^{\alphabold}_T\right)^\gamma.
\end{equation}
The relative risk aversion is measured by $1-\gamma$. It is convenient to use another measure $\delta$ which is also known as a distortion rate (see \cite{ZAR01})
\begin{equation}
\nonumber
\label{eq:delta_def}
\delta = \frac{1}{1-\gamma}, \quad  0 < \delta < \infty
\end{equation}
so the smaller $\delta$ is, the less risk averse the agent. The case $\gamma=0$ corresponds to the logarithmic utility function and the investor with $\gamma \rightarrow 1$ is a risk seeking investor.

\section{Main result}
\label{sec:MainResult}
\subsection{The Hamilton--Jacobi--Belman equation}
Our aim is to find the optimal control $\alphab^*(W^{\alphabold}_t, \Xt, t)$ and the value function $J(W^{\alphabold}_t, \Xt, t)$ as the functions of wealth $W^{\alphabold}_t$, prices $\Xt$ and time $t$. The Hamilton--Jacobi--Bellman equation is
\begin{equation}
\label{eq:HBJ_eq}
\sup_{\alphab}
\left( \left( \partial /\partial t  + \mathcal{L}  \right)J \right) =0.
\end{equation}
Here $\mathcal{L}$ is the infinitesimal generator of the wealth process $W^{\alphabold}_t$:
\begin{eqnarray}
\nonumber
\label{eq:W_gen}
\mathcal{L} =
\frac{\alphab^{\top} \Corrmatrix \alphab}{2}   \frac{\partial^2}{\partial w^2}
+ \alphab^{\top} \Corrmatrix \partialX  \frac{\partial}{\partial w}
+ \frac{\partialX^{\top} \Corrmatrix \partialX}{2}
- \alphab^{\top} \kappamat \Xvec \frac{\partial}{\partial w}
- \Xvec^{\top} \kappamat\partialX
\end{eqnarray}
and the first order optimality condition on the control $\alphab^*$ is
\begin{equation}
\label{eq:FirstOrderConditions}
\alphab^*(w, \Xvec, t) = \frac{J_w }{J_{ww}} \Corrmatrix^{-1} \kappamat \Xvec - \frac{\partialX J_w}{J_{ww}}.
\end{equation}
The operator $\partialX$ denotes a vector differential operator
\begin{equation}
\nonumber
\partialX = \left[ \frac{\partial}{\partial x_1}, \frac{\partial}{\partial x_2}, \hdots, \frac{\partial}{\partial x_n}\right]^{\top}
\end{equation}
for which we define the following operations for any vectors $\textbf{a} \in \mathbb{R}^{1\times n}$ and matrices $\textbf{A} \in \mathbb{R}^{n \times n}$:
\begin{equation}
\nonumber
\textbf{a}^{\top} \partialX= \sum_{i =1}^{n} \textbf{a}_i \frac{\partial}{\partial x_i},
\quad
\partialX^{\top} \textbf{A} \partialX = \sum_{i = 1} ^{n} \sum_{j = 1} ^{n} \textbf{A}_{ij} \frac{\partial^2}{\partial x_i \partial x_j}.
\end{equation}
Note that the first summand in the right-hand side of (\ref{eq:FirstOrderConditions}) is the myopic demand term corresponding to a static optimization problem while the second term hedges from changes in the investment opportunity set. For a log utility investor ($\gamma = 0$ or, equivalently,  $\delta  = 1$) the second term vanishes (see \cite{ME}.)

Substituting this condition into the equation (\ref{eq:HBJ_eq}) for the value function, we obtain a non-linear PDE which can be linearised by the distortion transformation (see \cite{ZAR01}):
\begin{equation}
\nonumber
J(w,\Xvec, t) = \frac{w^\gamma}{\gamma} f^{1/\delta} (\Xvec, t).
\end{equation}
Here the function $f(\Xvec, t)$ is a solution to the Cauchy problem for the parabolic PDE:
\begin{eqnarray}
\nonumber
\label{eq:PDE_f_simplified}
\frac{\partialX^{\top} \Corrmatrix \partialX }{2} f
- \frac{\delta + 1}{2} \Xvec^{\top} \kappamat \partialX f
-\frac{\delta - 1}{2} \partialX^{\top} f \kappamat \Xvec +
\frac{\delta(\delta - 1)}{2} \Xvec ^{\top} \kappamat
\Corrmatrix^{-1} \kappamat \Xvec  f + \frac{\partial f}{\partial t} &=& 0.
\\ \nonumber
f(\Xvec, T) &=& 1.
\end{eqnarray}
\subsection{Solution}
The main equation (\ref{eq:PDE_f_simplified}) can be reduced to the matrix Riccati ODE. The value function $J$ and the optimal control $\alphab^*$ have quasi-analytic representations via solutions to this ODE.  Using an ansatz similar to \cite{BRENDLE} and \cite{PAP02}, we prove that the value function $J$ is given by
	\begin{equation}
	\nonumber
	\label{eq:value_function_distortion_represenation}
	J(w, \Xvec, t) = \frac{w^\gamma}{\gamma}
	\cdot
	\exp\left\{ \int_0^{T-t}  \frac{\textbf{Tr}\left( \Amatrix(u) \Corrmatrix \right) }{\delta} du  \right\}
	\cdot
	\exp\left\{\frac{\Xvec^{\top} \Amatrix(T-t)\Xvec}{\delta}
	\right\}  	
	\end{equation}
	where $\textbf{Tr}$ denotes trace operator and the function $\Amatrix : \mathbb{R}^+ \rightarrow  \mathbb{R}^{n\times n} \times \mathbb{R}^+$ is a matrix function of inverse time $ \tau = T -t$:
	\begin{equation}
	\nonumber
	\Amatrix(\tau) =
	\begin{vmatrix}
	A_{11}(\tau) & A_{12}(\tau) & \hdots & A_{1n} (\tau) \\
	A_{21}(\tau) & A_{22}(\tau) & \hdots & A_{2n} (\tau) \\
	\vdots       &    \vdots    & \ddots & \vdots        \\
	A_{n1}(\tau) & A_{n2}(\tau) & \hdots & A_{nn} (\tau) \\
	\end{vmatrix}
	\end{equation}
	which is defined as a solution to the following matrix Ricatti equation:
	\begin{eqnarray}
	\label{eq:A_Ricatti_Eq_Original}
	\Amatrix'(\tau) &=& \R_{\Corrmatrix, \kappamat, \delta} \Amatrix
	\\ \nonumber
	\Amatrix(0) &=& \textbf{0}
	\end{eqnarray}
	with $\R_{\Corrmatrix, \kappamat, \delta}$ denoting the nonlinear operator
	\begin{eqnarray}
	\label{eq:RicOperatorDef}
	\R_{\Corrmatrix, \kappamat, \delta} \Amatrix &=&
	\frac{\left( \Amatrix^{\top} + \Amatrix\right) \Corrmatrix \left( \Amatrix^{\top} + \Amatrix\right)}{2}
	\\ \nonumber
	&-& \frac{\delta + 1}{2} \kappamat \left(\Amatrix^{\top} + \Amatrix \right)
	- \frac{\delta - 1}{2} \left(\Amatrix^{\top} + \Amatrix \right) \kappamat
	+ \frac{\delta(\delta -1)}{2} \kappamat \Corrmatrix^{-1} \kappamat
	\end{eqnarray}
	The optimal strategy $\alphab^*$ has the following representation:
	\begin{eqnarray}
	\label{eq:alpha_Sym}
	\alphab^*(w, \Xvec, t) = w  \left[-\delta  \Corrmatrix^{-1} \kappamat  +
	\Amatrix + \Amatrix^{\top}  \right] \Xvec.
	\end{eqnarray}
Introducing a new matrix $\Dmatrix$ as
\begin{equation}
	\nonumber
	\Dmatrix(\tau) = \delta \Corrmatrix^{-1} \kappamat
	 -\left(\Amatrix(\tau) + \Amatrix^{\top}(\tau) \right)
\end{equation}
	we get the following formula for optimal strategy $\alphab^*$:
	\begin{eqnarray}
	\label{eq:alpha_D}
	\alphab^*(w, \Xvec, t) = -w \Dmatrix(\tau) \Xvec
	\end{eqnarray}
	Matrix $\Dmatrix$ can be found directly from another Riccati ODE:
	\begin{eqnarray}
	\label{eq:Riccati_D}
	\Dmatrix'(\tau) &=&  -\Dmatrix^{\top} \Corrmatrix \Dmatrix + \delta\kappamat \Corrmatrix^{-1} \kappamat.
	\\ \nonumber
	\Dmatrix(0) &=& \delta \Corrmatrix^{-1} \kappamat.
	\end{eqnarray}	
If one only needs the optimal control it is sufficient to solve the simpler equation (\ref{eq:Riccati_D}). To find the value functions, one needs to solve the more complex system (\ref{eq:A_Ricatti_Eq_Original}.)

Optimality of the candidate control $\alpha^*$ can be verified using the same arguments as in \cite{PAP02} (see also \cite{DAV01} and \cite{DAV02}.)

\section{Analysis. Review of the one-dimensional case}
\label{sec:Analysis1D}
\subsection{The problem}
Before we analyse the multidimensional case, let us present a short review of the one-dimensional case, for more details see \cite{BG}. It is obtained from our problem by setting $n = 1$ in all formulas from Section \ref{sec:price_process}. To be more precise, we consider mean-reverting asset $X_t$ which follows an Orntein--Uhlenbeck process with zero mean and unit variance:
\begin{equation}
\nonumber
dX_t = -\kappa X_t dt + dB_t
\end{equation}
and the wealth process $W^{\alphabold}_t$ generated by the trading strategy $\alpha$:
\begin{equation}
\nonumber
	dW^{\alphabold}_t = \alpha_t dX_t.
\end{equation}
We are looking for the maximizer $\alpha^*$ of the expected utility over the terminal wealth $W^{\alphabold}_T$:
\begin{equation}
\nonumber
\alpha^* = \argmax_{\alpha} \left[ \E_t\left[U(W^{\alphabold}_T) \right] \right].
\end{equation}
\subsection{The structure of the optimal strategy}
The optimal control $\alpha^*$ can be expressed as
\[
\alpha^*(w,x,t) = -wD_\kappa(T-t)x,
\]
where the function $D_\kappa(\tau)$ is a solution to the following Riccati equation:
\begin{eqnarray}
\label{eq:D_1D}
D_\kappa' &=& -D_\kappa^2 + \delta k^2
\\ \nonumber
D_\kappa(0) &=& \delta \kappa.
\end{eqnarray}
This one-dimensional problem (\ref{eq:D_1D}) can be solved explicitly (this can be done via the substitution $\tau(D_\kappa) = D_\kappa^{-1}$). The function $D_\kappa(\tau)$ is a shifted and scaled sigmoid function of the inverse time $\tau = T - t$ :
\begin{eqnarray}
\nonumber
D_\kappa(\tau) =\kappa \sqrt{\delta}
\frac{\sqrt{\delta} \cosh{\kappa \sqrt{\delta} \tau} + \sinh{\kappa \sqrt{\delta} \tau}}{\sqrt{\delta} \sinh{\kappa \sqrt{\delta} \tau} + \cosh{\kappa \sqrt{\delta} \tau} }
\end{eqnarray}

It is worth to mention that for $\gamma < 0 $ the function $D_\kappa$ can be represented as
\begin{eqnarray}
\nonumber
D_\kappa(\tau) =\kappa \sqrt{\delta}  \tanh \left(\kappa \sqrt{\delta} \tau + \varphi \right), \quad
\tanh \varphi = \sqrt{\delta}
\end{eqnarray}

\begin{figure}
\begin{center}
\resizebox*{13cm}{!}{\includegraphics{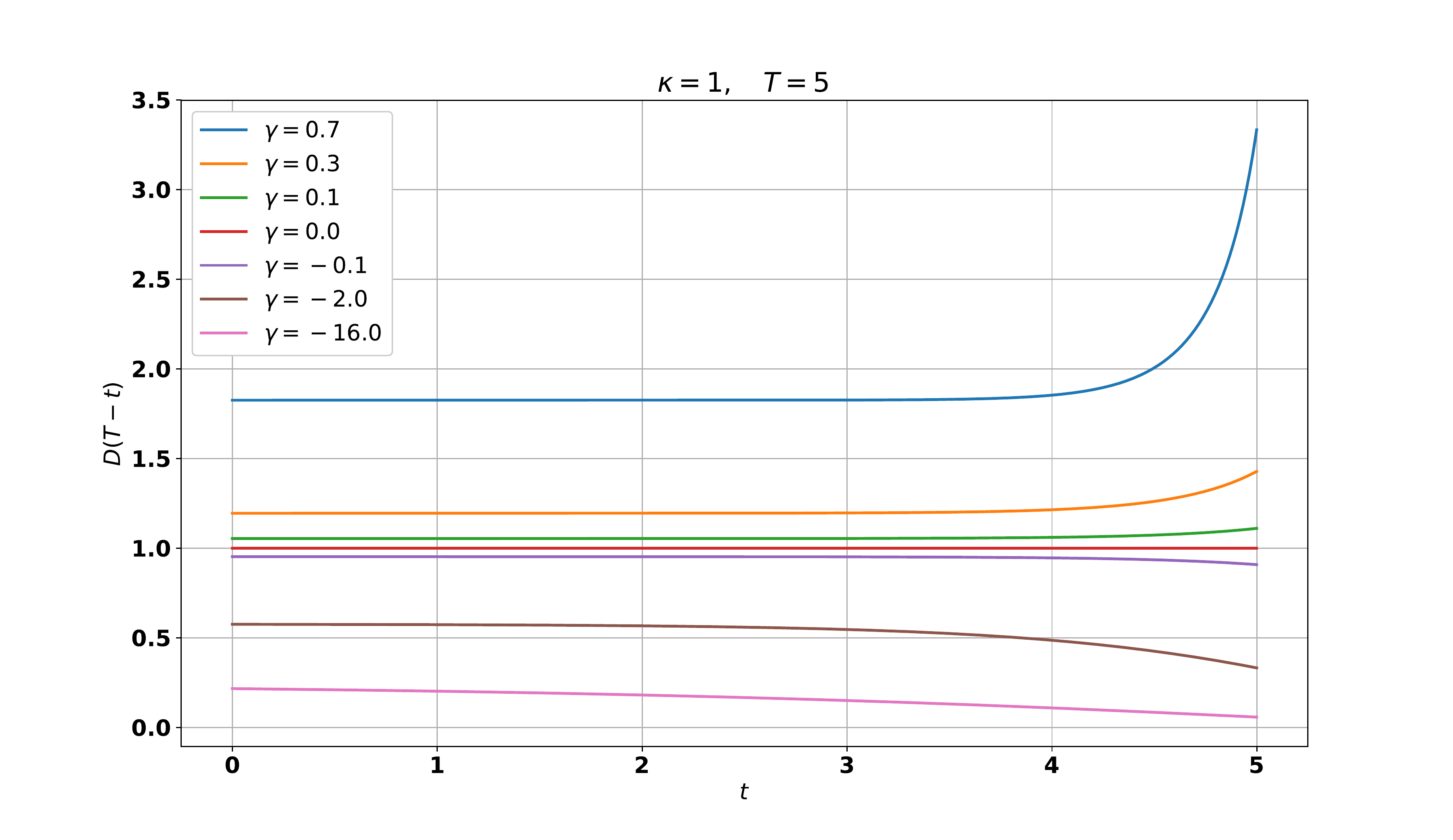}}
\caption{\label{fig:D1d}
Position size multiplier $D(T-t)$ for different values of risk aversion
}
\end{center}
\end{figure}

The behavior of the function $D_\kappa(T-t)$ depends on the value of risk aversion $\gamma$: an agent with negative gamma (less risk averse than log-utility inversor) becomes less agressive if time approaches to the terminal time while traders with positive gamma become more aggressive (see Figure \ref{fig:D1d}). For the log-utility agent ($\gamma = 0$, red line on Figure \ref{fig:D1d}) the optimal strategy is static, i.e.  $D_\kappa(\tau) \equiv const$.

\subsection{Value function structure}
The value function $J(w,x,t)$ can be split into three terms:
\begin{eqnarray}
\nonumber
J(w,x, t) = \underbrace{\frac{w^\gamma}{\gamma}}_{\textbf{a}}
\cdot
\underbrace{\exp\left\{ -\int_0^{T-t}  \frac{D(u) - \delta \kappa}{2\delta} du  \right\}}_{\textbf{b}}
\cdot
\underbrace{\exp\left\{ - \frac{x^2(D(T-t) - \delta \kappa)}{2\delta}
\right\} }_{\textbf{c}}
\end{eqnarray}
which can be interpreted as follows:
\begin{itemize}
\item \textbf{a}: present wealth utility,
\item \textbf{b}: time value (utility of future expected opportunities),
\item \textbf{c}: instrinsic value (utility of the immediate investment opportunity set.)
\end{itemize}

\subsection{Wealth process structure}
The stochastic process $W^{\alphabold}_t$ generated by the optimal strategy $\alphab^*$ can be represented as (for more details see \ref{app:SDE_sol})
\begin{eqnarray}
\nonumber
\label{eq:W_for_1D}
\log \left(\frac{W^{\alphabold}_t}{W^{\alphabold}_s}\right) = \underbrace{\int_{s}^t \frac{D_\kappa(T-u) - \delta \kappa^2 X_u^2}{2} du}_{\textbf{a}} + \underbrace{\frac{X_s^2 D_\kappa(T-s) - X_t^2 D_\kappa(T-t)}{2}}_{\textbf{b}}.
\end{eqnarray}
So the log return of wealth between times $s$ and $t$ is the sum of
\begin{itemize}
	\item \textbf{a}:  profit/loss from dynamic trading in the time period $[s, t]$,
	\item \textbf{b}:  profit/loss on position open at at time $s$.
\end{itemize}

\subsection{Monte Carlo simulations}
The higher mean reversion speed $\kappa$ makes trader more aggressive. Authors also make the following observations based on Monte Carlo simulations:
\begin{itemize}
	\item
	The influence of mean reversion coefficient misspecification is asymetric.
	\item
	Trading with a conservatively estimated $\kappa$ reduces greatly the utility uncertainty. The overestimation of $\kappa$ leads to excessively aggressive positions.
	It is much safer to underestimate $\kappa$ than to overestimate it.
\end{itemize}

\section{Analysis. Multidimensional case.}
\label{sec:AnalysisMD}
The main difference between multidimensional and one dimensional case is that changes in some spreads may affect positions in other spreads via changes in risk exposures. Generally, one might expect two possible motivations to take a position in each of the assets: to extract value from its reversion or to hedge positions in other assets.

In the multidimensional case, the time decay function $\Dmatrix$ is a matrix. The main difficulty is that there are no known techniques to explicitly solve generic matrix Riccati equations. However, there are several important special cases in which explicit solutions can be obtained. We start our analysis with these cases; based on these formulas we can demonstrate the main principles of interaction between asset prices and optimal positions.

For the rest of the paper, we will analyse only the case $\X_0 \equiv \thetavec$, i.e. the long-term investment behavior of the value function $J(w, \textbf{0}, t)$.
\subsection{Explicitly solvable cases.}
\subsubsection{Non-correlated assets}
Assume that the asset processes are driven by non-correlated Wiener processes, $\Corrmatrix = \Imatrix$. We can expect that the optimal strategy is simply a vector of one dimensional optimal strategies for each asset. That is, a candidate optimal control is
\begin{equation}
\nonumber
\alphab^* = -w \Dmatrix(\tau) \Xvec,
\quad
\Dmatrix (\tau)= diag \left( D_{\kappa_1}(\tau), D_{\kappa_2}(\tau) \hdots,D_{\kappa_n}(\tau) \right), \quad \tau = T - t.
\end{equation}
For the definition of $D_\kappa$ see \ref{sec:Analysis1D}. One can directly confirm that this control is indeed optimal by checking that it solves the system (\ref{eq:Riccati_D}).

In this case, there are no interactions between the assets. The position in the $i$-th assets depends only on time $t$, current wealth and $i$-th asset parameters.

\subsubsection{Common reversion rate}
Another case that allows an explicit solution is when the correlations are non-trivial but the reversion rate $\kappa$ is the same for all assets $\kappamat = \kappa \Imatrix$. Recall SDE for the price process
\begin{equation}
\nonumber
d\Xt = -\kappa \Xt dt + d\Bt, \quad d\Bt d\Bt^{\top} = \Theta dt.
\end{equation}
We show that for this case the explicit solution can also be constructed.

Indeed, with a single common reversion rate, any non-zero linear combination $\Yt = \Lmatrix^{-1} \Xt$ of Ornstein--Uhlenbeck processes is also an Ornstein--Uhlenbeck process:
\[
d\Yt = -\kappa \Yt dt + d\BtTilde, d\BtTilde d\BtTilde^{\top}  = \Lmatrix^{-1} \Theta (\Lmatrix^{-1})^{\top} dt
\]
Here $\BtTilde$ is a $n$- dimensional Wiener process with correlation matrix
\[
\Lmatrix^{-1} \Theta (\Lmatrix^{-1})^{\top}.
\]
Assuming invertibility of $\Lmatrix$, one can find an optimal control $\alphab_{Y}$ for this new process $\Yt$ and then transform it to an optimal control for $\Xt$. The transformation is based on the following equality
\[
dW^{\alphabold}_t= \alphab_{Y}^{\top} d\Yt = \alphab_{X}^{\top} d\Xt, \quad \alphab_{X} (W^{\alphabold}_t, \Xt, t) = (\Lmatrix^{-1})^{\top} \alphab_Y(W^{\alphabold}_t, \Lmatrix^{-1} \Xt, t).
\]
The transformaton matrix $\Lmatrix$ is constructed as a Cholesky decomposition of correlation matrix $\Corrmatrix$ :
\[
\Lmatrix^{\top} \Lmatrix = \Lmatrix \Lmatrix^{\top} = \Corrmatrix, \quad
(\Lmatrix^{-1})^{\top} \Lmatrix^{-1} = \Lmatrix^{\top} (\Lmatrix^{-1})^{\top} = \Corrmatrix^{-1}.
\]
Applying this transformation, we obtain the following equation for the optimal control:
\begin{eqnarray}
\nonumber
\alphab^* = -w D_\kappa(T-t) \Corrmatrix^{-1} \Xvec.
\end{eqnarray}
Thus, the optimal trading rule can be interpreted as constuction of linearly independent factor portfolios and then trading them in the manner of the previous case. This is similar to the portfolio signal construction approach of \cite{KMP}.

In this case, there are also no interactions between the assets. The value function $J(w \textbf{0}, t)$ does not depend on asset correlations:
\[
J(w, \textbf{0}, t) = \frac{w^\gamma}{\gamma} \exp\left\{ n\int_0^{T-t}\frac{\delta \kappa - D_\kappa (u)}{2\delta}  du\right\}
\]

\subsubsection{Hedging a mean reverting asset via correlated Brownian Motions}
\label{sec:kappa00}
Let us consider a case when the tradeable asset set consists of a single mean-reverting asset and one or several correlated Brownian motions. We can also consider this case as the limiting case for tradeable asset sets where one asset' mean reversion rate $\kappa$ is very large relatively to all other asset' reversion rates.

Consider the following matrix of reversion rates:
\[
\kappamat = diag(\kappa, 0, 0, \hdots, 0).
\]
One can check by a direct calculation that the solution to the Riccati equation (\ref{eq:Riccati_D}) has the following form:

\begin{equation}
\nonumber
\label{eq:D_matrix_zero_kappa}
\nonumber
\Dmatrix(t) =
\begin{vmatrix}
\Dmatrix_{11} & 0 & \hdots & 0 \\
\Dmatrix_{21} & 0 & \hdots & 0 \\
\vdots       &    \vdots    & \ddots & \vdots        \\
\Dmatrix_{n1} & 0 & \hdots & 0 \\
\end{vmatrix}
\end{equation}
$\Dmatrix_{j1} = \delta \kappa \left(\Corrmatrix^{-1}\right)_{j1}$.
The term $\Dmatrix_{11}(\tau)$ can be derived from the following Riccati ODE:
\begin{eqnarray}
\nonumber
\Dmatrix_{11}'(\tau) &=& -\Dmatrix_{11}^2 + 2\delta (\zeta - 1) \kappa \Dmatrix_{11}
+\kappa^2 \delta \zeta \left( \delta (1-\zeta) + 1\right)
\\ \nonumber
\Dmatrix_{11}'(0) &=& \delta \zeta \kappa.
\end{eqnarray}
This ODE can be solved explicitly to yield the following formula for $\Dmatrix$:
\begin{eqnarray}
\nonumber
\Dmatrix_{11}(\tau) =
	\left\{
	\begin{array}{l}
	\kappa \lambda
	\frac{\delta \cosh \lambda \kappa \tau + \lambda \sinh \lambda \kappa \tau}
	{\delta \sinh \lambda \kappa \tau +\lambda \cosh \lambda \kappa \tau}
	+\delta \kappa(\zeta - 1),
	\quad \gamma <  1 /\zeta
	\\
	\kappa \delta
	\frac{1}
	{1+ \delta \kappa \tau}
	+\delta \kappa(\zeta - 1),
	\quad \gamma =  1 / \zeta
	\\
	\kappa \lambda
	\frac{\delta \cos \lambda \kappa \tau -\lambda \sin \lambda \kappa \tau}
	{\delta \sin \lambda \kappa \tau +\lambda \cos \lambda \kappa \tau}
	+\delta \kappa(\zeta - 1),
	\quad  1 /\zeta < \gamma < 1.
	\end{array}
	\right.
\end{eqnarray}
Here
\[
\zeta =  \left(\Corrmatrix^{-1}\right)_{11}, \quad \lambda = \sqrt{|\delta (\delta - 1) \zeta -\delta ^2|}
\]

Thus, in this case we trade the mean-reverting asset and hedge it via correlated Brownian motions. Both the mean revertion asset position and the hedging positions are larger for large correlations. Availability of correlated hedging assets allows us to take larger positions for given risk aversion and wealth.

\subsection{The structure of the optimal strategy}
To illustrate the structure of the optimal strategy, we expand the product $\Dmatrix(\tau) \Xvec$ in formula (\ref{eq:alpha_D}) for optimal control $\alphab^*$:
\begin{equation}
\nonumber
\begin{vmatrix}
\alphab^*_1 \\
\alphab^*_2 \\
\vdots \\
\alphab^*_n
\end{vmatrix}
= -w
\begin{vmatrix}
\Dmatrix_{11}(\tau) \Xvec_1 + \Dmatrix_{12}(\tau) \Xvec_2 +\hdots \Dmatrix_{1n} (\tau)\Xvec_n \\
\Dmatrix_{21}(\tau) \Xvec_1 + \Dmatrix_{22} (\tau)\Xvec_2 +\hdots \Dmatrix_{2n} (\tau) \Xvec_n \\
\vdots \\
\Dmatrix_{n1}(\tau) \Xvec_1 + \Dmatrix_{n2}(\tau) \Xvec_2 +\hdots \Dmatrix_{nn} (\tau) \Xvec_n \\
\end{vmatrix}
\end{equation}
The summand $\Dmatrix_{ii} x_i$ is a position size multiplier for a mean reversion trading of $i-th$ asset while $\Dmatrix_{ij} x_j$ is a quantity of $i-th$ asset required to hedge the position in $j-th$ asset. In case of non-correlated assets each $\Dmatrix_{ij} = 0$, for $i \neq j$. The quantities  $\Dmatrix_{ij}$ and $\Dmatrix_{ji}$ satisfy the following relations :
\begin{eqnarray}
\nonumber
\Dmatrix_{ij} +   \delta \Corrmatrix^{-1}_{ij} \kappa_j= \Dmatrix_{ji} + \delta \Corrmatrix^{-1}_{ij} \kappa_i.
\end{eqnarray}
Note that the difference between $\Dmatrix_{ij}$ and $\Dmatrix_{ji}$ does not depend on time $t$.

\subsection{Wealth dynamics}
Similarly to the one-dimensional case, the wealth process $W^{\alphabold}_t$ can be expressed as
\begin{eqnarray}
\label{eq:W_for_nD}
\begin{array}{c}
\log \left( \frac{W^{\alphabold}_t}{W^{\alphabold}_s}\right) = \overbrace{\int_{s}^t \frac{\textbf{Tr}\Corrmatrix \Dmatrix(T-u) -\delta \X_u^{\top} \kappamat \Corrmatrix^{-1} \kappamat \X_u
} {2}du }^{\textbf{a}}
+
\\
+
\underbrace{\frac{\X_s^{\top} \Dmatrix(T-s) \X_s - \X_t^{\top} \Dmatrix(T-t) \X_t}{2}}_{\textbf{b}}
+
\underbrace{\frac{1}{2} \int_{s}^t \X_u^{\top} \left[ \Dmatrix - \Dmatrix^\top\right]  d\X_u}_{\textbf{c}}
\end{array}
\end{eqnarray}

One term of equation (\ref{eq:W_for_nD}) that is missing in the one-dimensional case is $\textbf{c}$. This summand corresponds to hedging efficiency. It is easy to see that for cases $\Corrmatrix = \Imatrix$ or $\kappamat =\kappa \Imatrix$ this term vanises. As we mentined before, the case $\kappamat =\kappa \Imatrix$ can be reduced to the case $\Corrmatrix = \Imatrix$.
\subsection{Example. 2-dimensional model}.
\label{sec:example2d}

%

To illustrate interactions between reversion speed and correlation, let us consider a two-dimensional example in more details. We will use the following parameters for this illustration: numbers of assets be $n = 2$, noise magnitude $\sigmamat = \Imatrix$, long term mean and initial point $\thetavec = \X_0 = 0$, risk aversion $\gamma = -4$ and time horizon $T = 3$. We consider an optimal strategy for a portfolio of two correlated Ornstein--Uhlenbeck processes with $\kappa_1 = 1$ and different values of $\kappa_2$ and correlation $\rho$.
\[
n = 2,\quad \gamma = -4, \quad \sigmamat = \Imatrix, \quad \kappamat = diag(1,\kappa_2),
\quad
\thetavec = \X_0 = \textbf{0},
\quad
\Corrmatrix = \begin{bmatrix}
1 & \rho\\
\rho & 1
\end{bmatrix}
\]

Figure \ref{fig:JkappaRho} shows of the value function $J$ as a function of $\log (\kappa_2 / \kappa_1)$ ($\kappa_1 = 1$) for several different values of $\rho$. We are varying here the lower of two asset mean-reversion rates. It turns out that for sufficiently high correlation $\rho$, the value function has a proper minima as function of $\kappa_2$ and it becomes decreasing in $\kappa_2$ as correlation gets closer to $1$. This means that in these cases, one would prefer to have a lower value for the second asset'  mean-reversion rate to a slightly higher value (but not to a much higher value $\kappa_2>>\kappa_1$. Therefore, with more that one asset, a higher reversion rate is not always good for extracting value from trading, quite unlike the one-dimensional case.
\begin{figure}
	\begin{center}
		\resizebox*{13cm}{!}{\includegraphics{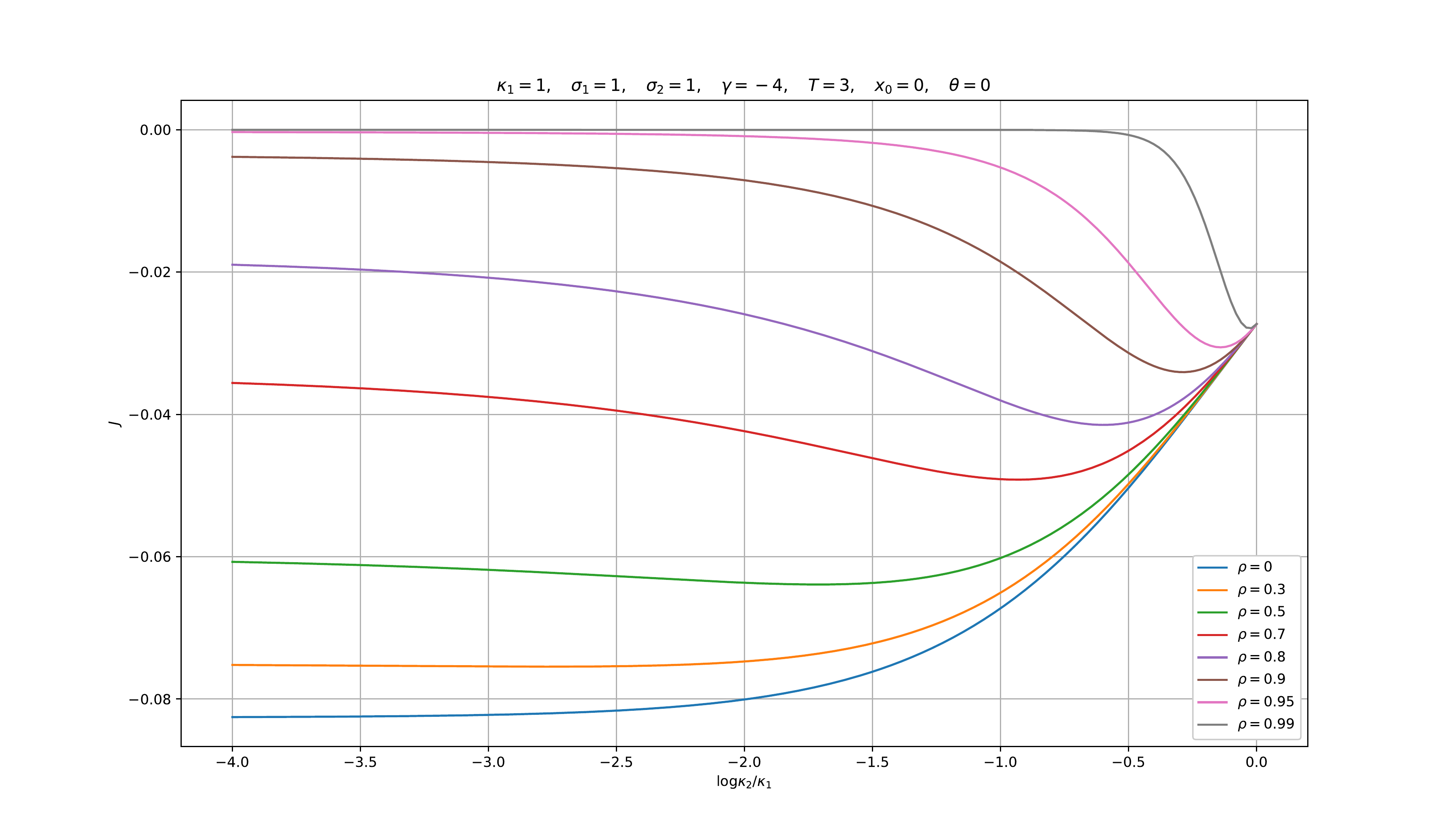}}
		\caption{
			\label{fig:JkappaRho}
			2D example. Value function for a range of values for $\kappa_2$ and correlation $\rho$.
		}
	\end{center}
\end{figure}

\subsection{Impact of correlation}
\label{sec:correlation}
We have seen in the previous section that the value function can be non-monotonic in mean-reversion rates. Let us show that it is always increasing with the correlation all other parameters being equal.

Suppose now that we start our trading process with no immediate trading opportunities (i.e. $\Xvec = \textbf{0}$). We consider $J(w, \textbf{0}, t)$ as the function on correlation coefficients $\rhos$. In the standard Markowitz portfolio optimization problem, one can construct more profitable portfolios when correlations are lower. In our setting, we can prove that the value function has a local minima at zero correlations $\Corrmatrix = \Imatrix$. Correlations between driving processes enable cross-hedging between positions in different assets and these increase the value function. We have already seen a similar beneficial effect of higher correlations in section \ref{sec:kappa00} for a special case of a single mean-reverting asset hedged with Brownian motions and the following theorem demonstrates that this effect holds in the general case as well.

\begin{theorem}
\label{thm:minima_corr}
In the absense of immediate trading opportunities ($\Xvec = \textbf{0}$) the value function $J(w, \textbf{0}, t)$ as a function of pairwise correlation coefficients $\rhos$ has a local minima at $\Corrmatrix = \Imatrix$.
\end{theorem}
\begin{proof}
	Recall the representation of the value function:
	\[
	J(w,\textbf{0}, t) = \frac{w^\gamma}{\gamma}
	\exp \left\{
	\frac{1}{\delta} \int_0^{T-t} \textbf{Tr} \left(\Fmatrix(u)\right) du
	\right\}
	\]
	where matrix $\Fmatrix$ is equal
	\begin{equation}
	\label{eq:F_def}
	\Fmatrix = \frac{1}{2}(\Amatrix + \Amatrix^{\top})\Corrmatrix.
	\end{equation}
	Define new matrix $\Gammamat$ :
	\begin{equation}
	\label{eq:Gamma_def}
	\Gammamat = \Corrmatrix^{-1} \kappamat \Corrmatrix
	\end{equation}
	Note that $\Gammamat$ is a result of similarity transformation of the matrix $\kappamat$ and $\limrhoI \Gammamat = \kappamat$.  For the matrix $\Fmatrix$ we have the following ODE:
	\begin{eqnarray}
	\Fmatrix' &=& 2 \Fmatrix^2 - \delta \left(\kappamat \Fmatrix + \Fmatrix \Gammamat \right) + \frac{\delta(\delta - 1)}{2} \kappamat \Gammamat
	\\ \nonumber
	\Fmatrix(0) &=& \textbf{0}.
	\end{eqnarray}
	Let $\rhos$ be an arbitrary correlation coefficient at the position $\corridx$ (i.e. $\corridx = (ij)$, $\Corrmatrix_{ij} =\Corrmatrix_{ji}= \rhos$) and let us consider the following partial derivatives:
	\begin{eqnarray}
	\nonumber
	\frac{\partial J(w,\textbf{0}, t)}{\partial \rhos} &=&
	\frac{J(w, \textbf{0}, t)}{\delta}
	\int_0^{T-t} \textbf{Tr}
	\left( \frac{\partial \Fmatrix(u)}{\partial \rhos}\right) du
	\\ \nonumber
	\frac{\partial^2 J(w,\textbf{0}, t)}{\partial \rhos \partial \rhok} &=&
	\frac{J(w, \textbf{0}, t)}{\delta}
	\int_0^{T-t} \textbf{Tr}
	\left( \frac{\partial^2 \Fmatrix(u)}{\partial \rhos \partial \rhok}\right) du
	\\ \nonumber
	\frac{\partial^2 J(w,\textbf{0}, t)}{\partial \rhos^2} &=&
	\frac{J(w, \textbf{0}, t)}{\delta}
	\int_0^{T-t} \textbf{Tr}
	\left( \frac{\partial^2 \Fmatrix(u)}{\partial \rhos^2}\right) du
	\end{eqnarray}
	
	We will prove the following properties for any $\corridx$ and $\corridxsecond$:
	\begin{eqnarray}
	\label{eq:dJdrho_grad}
	\limrhoI \frac{\partial J(w,\textbf{0}, t)}{\partial \rhos} &=& 0
	\\
	\label{eq:dJdrho_mixed}
	\limrhoI \frac{\partial^2 J(w,\textbf{0}, t)}{\partial \rhos \partial \rhok} &=& 0
	\\
	\label{eq:dJdrho_second_sign}
	sign \limrhoI \frac{\partial^2 J(w,\textbf{0}, t)}{\partial \rhos^2} &=& sign \gamma,
	\quad (\kappa_{i} \neq \kappa_{j})
	\\
	\label{eq:dJdrho_second_same_kappa}
	\limrhoI \frac{\partial^2 J(w,\textbf{0}, t)}{\partial \rhos^2} &=& 0 \quad (\kappa_{i} = \kappa_{j})
	\end{eqnarray}
%
From equation (\ref{eq:dJdrho_grad}), the point $\Corrmatrix = \Imatrix$ is an extrema point. Equation (\ref{eq:dJdrho_mixed}) implies that the Gessian matrix at  $\Corrmatrix = \Imatrix$ is a diagonal matrix.	Using Silvester's criterion we prove that Gessian matrix is a positive definite at the point $\Corrmatrix = \Imatrix$, for more details see Appendix \ref{app:TH1_proof}.
\end{proof}

\section{Wealth distribution moments and analysis of parameters mis-specification}
\label{sec:MISSPEC}
\subsection{Closed from formulas.}
\label{sec:AQ_system}
In practice, one does not know the true values for model parameters, so it is important to understand value function sensitivities to errors in parameters estimation. In this section, we present an ODE based framework for the analysis of parameter mis-specification sensitivity. We provide semi-explicit formulas for the value function corresponding to misspecified parameters.
Let $\kappamatEst, \sigmamatEst, \CorrmatrixEst$ be an estimates of reversion rates, volatility and correlation. We consider the control $\alphabEst$ as a function of these estimates
\[
\alphabEst = w\sigmamatEst^{-1} \left[ - \delta \CorrmatrixEst^{-1} \kappamatEst + \left(\AmatrixEst^{\top} + \AmatrixEst \right) \right] \sigmamatEst^{-1} \Xvec.
\]
Here the matrix $\AmatrixEst$ is a solution to the following ODE
\begin{eqnarray}
\label{eq:A_Ricatti_Eq_Est}
\AmatrixEst'(\tau) &=& \R_{\CorrmatrixEst, \kappamatEst, \delta} \AmatrixEst
\\ \nonumber
\AmatrixEst(0) &=& \textbf{0},
\end{eqnarray}
where the differential operator $\R$ is defined in (\ref{eq:RicOperatorDef}). The wealth process $\hat{W}_t$ generated by the strategy $\alphabEst$ is a solution to the following SDE
\begin{equation}
\label{eq:W_hat_def}
d \hat{W}_t = \alphabEst^{\top}_t d\Xt
\end{equation}

\begin{theorem}
	Let $P_{\epsilon}(w, \Xvec, t)$ be the following expectation of a function of terminal wealth $\hat{W}_T$ defined by (\ref{eq:W_hat_def}):
	\begin{eqnarray}
	\nonumber
	P_{\epsilon}(w, \Xvec, t) = \mathbb{E} \left[  \frac{\hat{W}_T^\epsilon }{\epsilon}  \,\Big|\, \hat{W}_t = w, \, \Xt = \Xvec \right].
	\end{eqnarray}
	The expectation $P_{\epsilon}(w, \Xvec, t)$ can be explicitly found in the following form
	\begin{eqnarray}
	\label{eq:Pdef}
	P_{\epsilon}(w, \Xvec, t) &=& \frac{w^{\epsilon}}{\epsilon}
	\cdot
	\exp\left\{
	\int_{0}^{T-t} \textbf{Tr}\left(\Corrmatrix \Qmatrix(u)\right) du
	\right\}
	\\ \nonumber
	&\cdot&
	\exp\left\{\Xvec^{\top} \sigmamat^{-1}\Qmatrix(T-t) \sigmamat^{-1} \Xvec
	\right\},
	\end{eqnarray}
	where matrix $\Qmatrix$ is a solution to Riccati equation
	\begin{eqnarray}
	\label{eq:Q_eq}
	\Qmatrix' &=& \Qeq \Qmatrix
	\\ \nonumber
	\Qmatrix(0) &=& \textbf{0}.
	\end{eqnarray}
	The nonlinear operator $\Qeq$ is given by
	\begin{eqnarray}
	\nonumber
		\Qeq \Qmatrix  &=&
	\frac{\left(\Qmatrix + \Qmatrix^{\top}\right)\Corrmatrix \left(\Qmatrix + \Qmatrix^{\top}\right)}{2}  +
	\\  \nonumber
	&+&
	\left( \epsilon \betamat^{\top} \Corrmatrix - \kappamat \right)\left(\Qmatrix + \Qmatrix^{\top}\right)
	+ \frac{\epsilon (\epsilon -1)}{2} \betamat^{\top} \Corrmatrix \betamat
	- \epsilon \betamat^{\top} \kappamat
	\end{eqnarray}	
	and the matrix $\betamat$ is defined as
	\begin{equation}
	\nonumber
	\label{eq:beta_def}
	\betamat = \sigmamat \sigmamatEst^{-1} \left[-\delta \CorrmatrixEst^{-1} \kappamatEst +
	\left(\AmatrixEst + \AmatrixEst^{\top} \right)\right] \sigmamatEst^{-1} \sigmamat
	\end{equation}
	here the matrix $\AmatrixEst$ is a solution to the equation (\ref{eq:A_Ricatti_Eq_Est}).
\end{theorem}
In the setting $\epsilon = \gamma$ we obtain the expected utility corresponding to the misspecified parametes. The values $\epsilon = 1$ or $\epsilon = 2$  corresponds to the first two moments of $W_T$, so we can calculate Sharpe ratio:
\begin{equation}
\nonumber
Sh[\alphabEst] = \frac{P_{1}(w, \Xvec, t)}{\sqrt{2P_2(w,\Xvec, t) - P_1^2(w, \Xvec, t)}}.
\end{equation}
It is worth to mention, that the effects on misspecified long term mean level $\thetavec$ can be also analysed in the same way. For this case, we have to add extra term
\[
\exp\left\{ \Xvec^{\top} \textbf{V}\right\}
\]
to the equation (\ref{eq:Pdef}). Here $\textbf{V}$ is an $n \times 1$ vector function of inverse time $T - t$.

As an alternative, one can analyse the effect of parameter misspecification by using Monte-Carlo methods. However, from our point of view, the proposed ODE approach is computationally much more efficient than Monte-Carlo simulations.

\subsection{Impact of mis-specified reversion rates}
\begin{figure}
	\centering
	\begin{subfigure}
	\centering
	\includegraphics[height=2.2in]{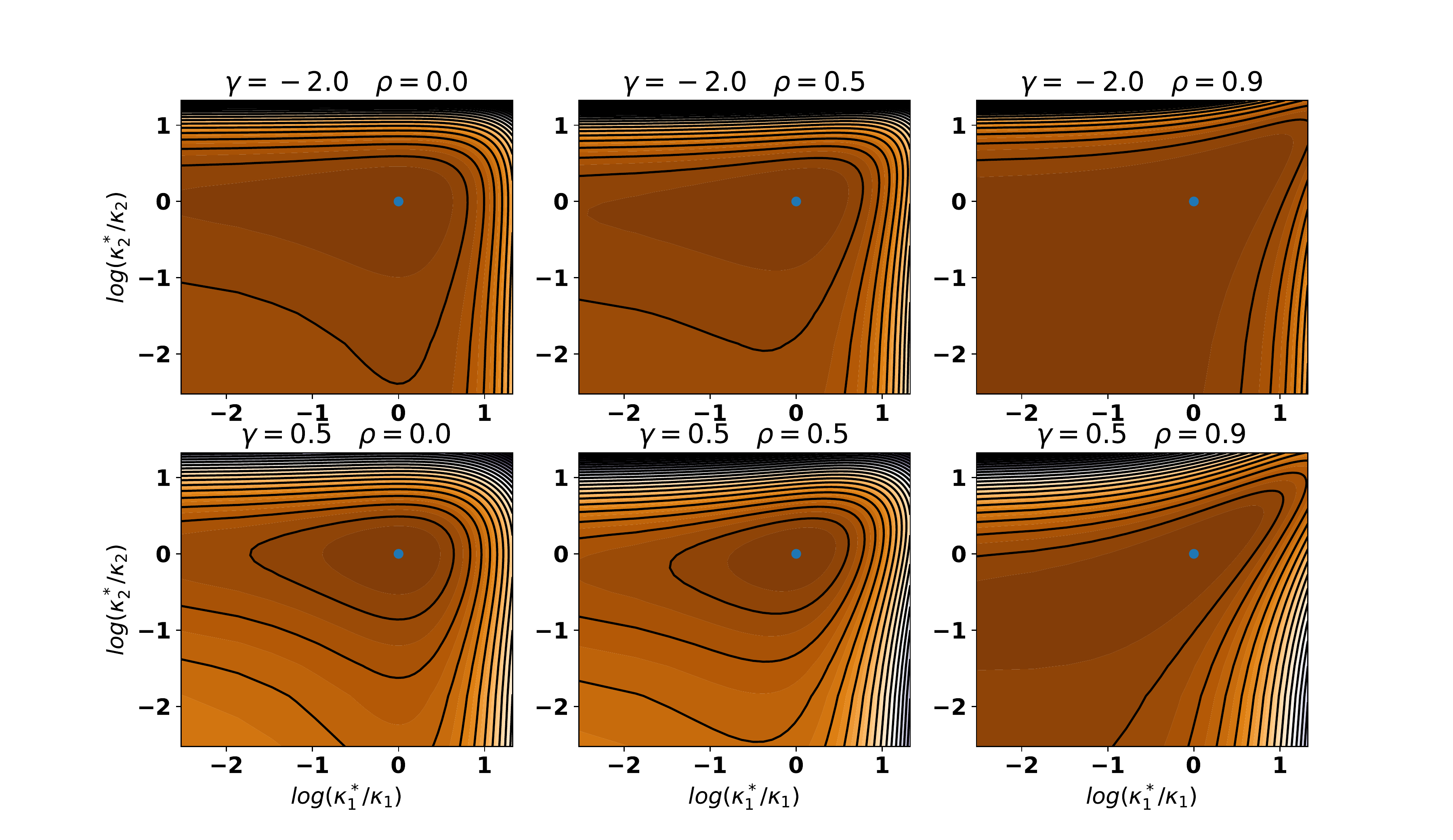}
	\end{subfigure}
~
	\begin{subfigure}
	\centering
	\includegraphics[height=2.in]{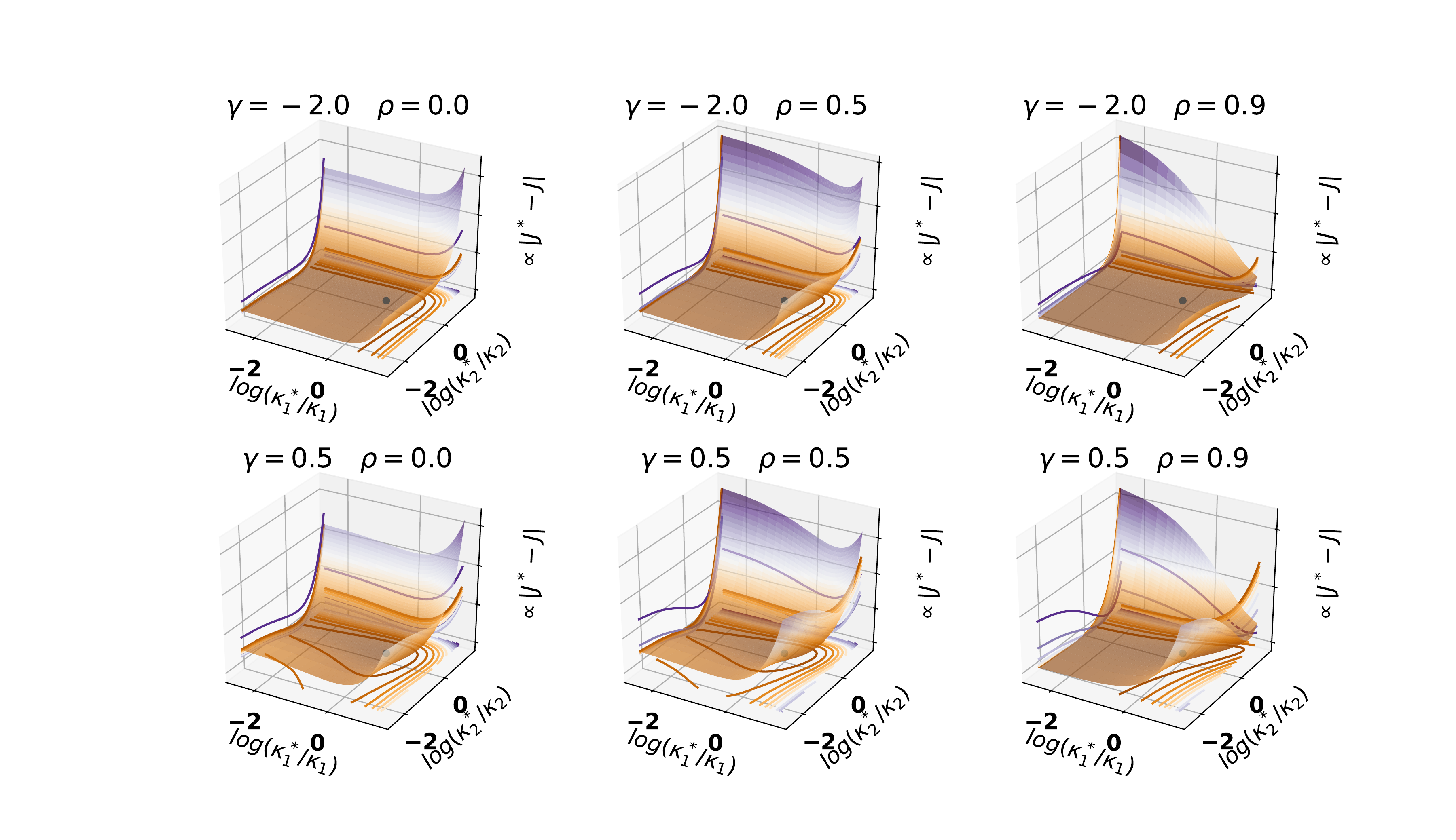}
	\end{subfigure}
\caption{Misspecified reversion rates.  Heatmap plot and 3D plot.}
		\label{fig:kappa_ms}
\end{figure}

We illustrate the method presented above on the analysis of misspecified reversion rates $\kappamat$. For simplicity, we consider the portfolios with only two assets. The results are presented on figure \ref{fig:kappa_ms}. We measure effect on misspecification by the difference between the value functions corresponding to true and mis-specified parameters (color and value of z-axis respectively).

Similarly to the one-dimensional case, the infuence of mean reversion coefficient misspecification is asymmetric.
Depending on the value of correlation, correct estimation of the ratio between reversion rates is more important than the estimations of the  exact values of each mean-reversion rate. It follows from the nature of optimal strategy: the faster mean-reverting asset is hedged in the slower one and the hedging accuracy depends on the ratio between reversion speeds.

\section*{Acknowledgments}
Dmitry Muravey acknowledges support by the Russian Science Foundation under the Grant number 20-68-47030.

\appendix

\section{Reduction of the HJB equation to the linear PDE}
\subsection{Distortion transformation}
\label{sec:Reduction}
The first order optimality condition on the control $\alphab^*$ yields the following linear system for the $\alphab^*$:
\begin{equation}
J_{ww} \Corrmatrix\alphab^* = \kappamat \Xvec J_w -\Corrmatrix \sigmamat \partialX J_{w}.
\end{equation}
The solution of this system reads
\begin{equation}
\alphab^* = \frac{1}{J_{ww}}  \left( \Corrmatrix^{-1} \kappamat \Xvec -  \partialX\right)J_w
\end{equation}
Using again the first order optimality condition, we get:
\[
(\alphab^*)^{\top} \kappamat \Xvec J_w - (\alphab^*)^{\top}  \Corrmatrix \partialX J_{w} = (\alphab^*)^{\top} \Corrmatrix  \alphab^* J_{ww}
\]
Substituting it into HJB equation we arrive at the following terminal problem for PDE:
\begin{eqnarray}
J_t  - \frac{1}{2}  \left(\alphab^*\right)^{\top} \Corrmatrix  \alphab^* J_{ww}
-
\Xvec^{\top} \kappamat \partialX J +  \frac{1}{2} \partialX^{\top}  \Corrmatrix \partialX J &=&0,
\\ \nonumber
J(w,\Xvec, T) &=& \frac{w^\gamma}{\gamma}
\end{eqnarray}
Plugging the exact value for an optimal control $\alphab^*$ yields non-linear PDE:
\begin{eqnarray}
\nonumber
J_t - \frac{1}{2} \frac{J_{w}^2}{ J_{ww}}
\left(\kappamat \Xvec \right)^{\top} \Corrmatrix^{-1} \left(\kappamat \Xvec \right) +
\frac{1}{2}\frac{J_{w}}{J_{ww}}
\left[
\left(\kappamat \Xvec \right)^{\top} \partialX J_w
+ \partialX^{\top} J_w  \left(\kappamat \Xvec \right)\right]
\\ \nonumber
- \frac{1}{2} \frac{1}{J_{ww}}  \partialX^{\top} J_w \Corrmatrix  \partialX J_w
-
\Xvec^{\top} \kappamat \partialX J +  \frac{1}{2} \partialX^{\top} \Corrmatrix \partialX J =0.
\end{eqnarray}
We proceed with an application of the so-called distortion transformation:
\begin{eqnarray}
J = \frac{w^\gamma}{\gamma} f^{1/\delta}(x,t), \quad \delta = \frac{1}{1-\gamma}
\end{eqnarray}
The exact formulas for the partial derivatives of the value function $J$ reads
\begin{eqnarray}
\nonumber
J_t &=& \frac{1}{\delta} \frac{J}{f} \frac{\partial f}{\partial t}, \quad
J_w = \frac{\gamma}{w} J, \quad
J_{ww} = \frac{\gamma(\gamma - 1) }{w^2} J
\\ \nonumber
\partialX J &=& \frac{1}{\delta} \frac{J }{f}  \partialX f, \quad
\partialX J_w =
\frac{\gamma}{w}  \frac{1}{\delta} \frac{J }{f}  \partialX f
\end{eqnarray}
Plugging in these expressions into terms of non-linear HJB PDE we get:
\begin{eqnarray}
\nonumber
\frac{1}{2} \partialX^{\top} \Corrmatrix  \partialX J
&=&
\frac{1}{2}\frac{1}{\delta} \frac{J}{f} \partialX^{\top} \Corrmatrix  \partialX f +
\frac{1}{2} \frac{1}{\delta} \left(\frac{1}{\delta} - 1\right) \frac{J}{f^2} \partialX^{\top} f  \Corrmatrix  \partialX f.
\\ \nonumber
-\frac{1}{2} \frac{1}{J_{ww}}  \partialX^{\top} J_w \Corrmatrix  \partialX J_w &=&
-\frac{1}{2} \frac{\gamma^2}{w^2} \frac{1}{\delta ^2}  \frac{J^2}{f^2} \frac{w^2}{\gamma (\gamma - 1) J}
\partialX^{\top} f \Corrmatrix  \partialX f
\\ \nonumber
&=&
\frac{1}{2} \frac{\gamma}{\delta} \frac{J}{f^2}
\partialX^{\top} f  \Corrmatrix  \partialX f
\\ \nonumber
&=&
-\frac{1}{2} \frac{1}{\delta} \left(\frac{1}{\delta} - 1 \right) \frac{J}{f^2}
\partialX^{\top} f  \Corrmatrix  \partialX f
\\ \nonumber
- \frac{1}{2} \frac{J_{w}^2}{ J_{ww}} &=&
- \frac{1}{2} \frac{\gamma^2}{w^2} J^2 \frac{w^2}{\gamma (\gamma - 1) J}
\\ \nonumber
&=&
\frac{1}{2} \frac{\gamma}{1-\gamma} J
\\ \nonumber
&=& \frac{1}{2} \frac{1}{\delta} \delta (\delta - 1) J
\\ \nonumber
\frac{1}{2}\frac{J_{w}}{J_{ww}}
\left[
\left(\kappamat \Xvec \right)^{\top} \partialX J_w
+ \partialX^{\top} J_w  \left(\kappamat \Xvec \right)\right]
&=&
\frac{1}{2} \frac{\gamma J}{w} \frac{w^2}{\gamma (\gamma - 1) J}
\Bigg[
\left(\kappamat \Xvec \right)^{\top} \left(\frac{\gamma}{w}  \frac{1}{\delta} \frac{J }{f}  \partialX f \right)
\\ \nonumber
&+& \left( \frac{\gamma}{w}  \frac{1}{\delta} \frac{J }{f}  \partialX f\right)^{\top} \left(\kappamat \Xvec \right)\Bigg]
\\ \nonumber
&=&
\frac{1}{2} \frac{1}{\delta} \frac{\gamma}{ \gamma - 1} \frac{J}{f}
\left[\Xvec^{\top} \kappamat \partialX f + \partialX^{\top} f \kappamat \Xvec \right]
\\ \nonumber
&=&
\frac{1- \delta}{2} \frac{1}{\delta} \frac{J}{f}
\left[\Xvec^{\top} \kappamat \partialX f + \partialX^{\top} f \kappamat \Xvec \right]
\\ \nonumber
&=&
-\frac{\delta - 1}{2} \frac{1}{\delta} \frac{J}{f}
\left[\Xvec^{\top} \kappamat \partialX f + \partialX^{\top} f \kappamat \Xvec \right]
\end{eqnarray}
This yields the following linear equation for the function $f$:
\begin{eqnarray}
\nonumber
\label{eq:PDE_f}
\frac{1}{2} \partialX \Corrmatrix \partialX f
- \frac{\delta + 1}{2} \Xvec^{\top}  \kappamat  \partialX f
-\frac{\delta - 1}{2} \partialX^{\top} f \left(\kappamat \Xvec \right) +
\frac{1}{2}\delta(\delta - 1) \left(\kappamat \Xvec \right)^{\top}
\Corrmatrix^{-1} \left(\kappamat \Xvec \right) f + \frac{\partial f}{\partial t} &=& 0.
\end{eqnarray}
or
\begin{eqnarray}
\nonumber
\label{eq:PDE_f_siml}
\frac{1}{2} \partialX \Corrmatrix  \partialX f
- \frac{\delta + 1}{2} \Xvec^{\top} \kappamat \partialX f
-\frac{\delta - 1}{2} \partialX^{\top} f \kappamat \Xvec +
\frac{\delta(\delta - 1)}{2} \Xvec ^{\top} \kappamat
\Corrmatrix^{-1} \kappamat \Xvec  f + \frac{\partial f}{\partial t} &=& 0.
\end{eqnarray}

The optimal control $\alphab^*$ reads:
\begin{equation}
\alphab^*(w,\Xvec, t) = w \left[ - \delta \left( \sigmamat\Corrmatrix \sigmamat \right)^{-1} \kappamat \Xvec  + \frac{\partialX f}{f}\right].
\end{equation}


\section{Wealth SDE solution}
\label{app:SDE_sol}
The wealth process corresponding to the optimal control takes the following form :
\begin{eqnarray}
\nonumber
dW_t = -W_t \Xt^{\top} \Dmatrix^{\top} d\Xt.
\end{eqnarray}
We represent the process $W_t$ in the stochastic exponent form:
\[
W_t = W_0 e^{\lambdamat^{\top} \Yt}, \quad d\Yt = u dt + \etamat d\Xt.
\]
and apply It\^o's lemma :
\begin{eqnarray}
\nonumber
dW_t = W_t \left[ \lambdamat^{\top} d\Yt + \frac{1}{2} \lambdamat^{\top} d\Yt d\Yt^{\top} \lambdamat\right].
\end{eqnarray}
Let us note that
\begin{eqnarray}
\nonumber
\lambdamat^{\top} u  &=& -\frac{1}{2} \lambdamat^{\top} \etamat \Corrmatrix \etamat^{\top} \lambdamat
\\ \nonumber
\lambdamat^{\top} \etamat &=& -\Xt^{\top}\Dmatrix^{\top}
\\ \nonumber
\etamat^{\top} \lambdamat  &=& -\Dmatrix \Xt
\\ \nonumber
\lambdamat^{\top} u &=& -\frac{1}{2} \Xt^{\top} \Dmatrix^{\top} \Corrmatrix \Dmatrix \Xt
\\ \nonumber
\lambdamat^{\top} d\Yt &=& \lambdamat^{\top} u dt + \lambdamat \etamat d\Yt.
\\ \nonumber
\lambdamat^{\top} d\Yt &=& -\frac{1}{2} \Xt^{\top} \Dmatrix^{\top} \Corrmatrix \Dmatrix \Xt dt  -\Xt^{\top}\Dmatrix^{\top} d\Xt
\end{eqnarray}
Therefore
\begin{eqnarray}
\nonumber
\int_0^t \lambdamat^{\top} d\Y_s &=& -\frac{1}{2} \int_0^t \X_s^{\top} \Dmatrix(T-s)^{\top} \Corrmatrix \Dmatrix(T-s) \X_s ds
-
\int_0^t \X_s^{\top}\Dmatrix(T-s)^{\top} d\X_s
\end{eqnarray}
Using that the matrix $\Dmatrix$  solves the following Riccati ODE:
\[
-\frac{d\Dmatrix}{dt} = \Dmatrix^{\top} \Corrmatrix \Dmatrix - \delta \kappamat \Corrmatrix^{-1} \kappamat
\]
we get
\begin{eqnarray}
\nonumber
W_t &=& W_0
\exp\left\{-\frac{\delta}{2} \int_0^t \X_s^{\top} \kappamat \Corrmatrix^{-1} \kappamat \X_s ds -\frac{1}{2} \left[\X_t^{\top} \Dmatrix(T-t) \X_t - \X_0^{\top} \Dmatrix(T) \X_0\right]\right\}
\\ \nonumber
&\cdot&
\exp\left\{ \frac{1}{2} \int_{0}^t \textbf{Tr}\Corrmatrix \Dmatrix(T-s) ds+
\frac{1}{2} \int_{0}^t \X_s^{\top} \left[ \Dmatrix - \Dmatrix^{\top}\right] d\X_s
\right\}
\end{eqnarray}

\section{Proof of Theorem \ref{thm:minima_corr}}
\label{app:TH1_proof}
The proof of Theorem \ref{thm:minima_corr} is equivalent to proof of the following 4 facts about matrix $\Fmatrix$:
\label{lem:D_properties}
\begin{eqnarray}
\label{lem:dDdp_0}
\limrhoI \left(\frac{\partial \Fmatrix}{\partial \rho_\corridx}\right)_{ij} & =& 0, \quad (ij) \notin \corridx.
\\ \label{lem:dDdp_1}
\limrhoI \textbf{Tr} \frac{\partial \Fmatrix}{\partial \rho_\corridx} & = & 0
\\ \label{lem:dDdp_2}
\limrhoI \textbf{Tr} \frac{\partial^2 \Fmatrix}{\partial \rhos \partial \rhok} &\equiv& 0.
\\ \label{lem:dDdp_3}
\limrhoI \textbf{Tr} \frac{\partial^2 \Fmatrix}{\partial \rhos ^2} &>& 0, \quad \gamma > 0,
\quad \kappa_i \neq \kappa_j.
\\ \nonumber
\limrhoI \textbf{Tr} \frac{\partial^2 \Fmatrix}{\partial \rhos ^2} &<& 0, \quad \gamma < 0,
\quad  \kappa_i \neq \kappa_j,
\\ \nonumber
\limrhoI \textbf{Tr} \frac{\partial^2 \Fmatrix}{\partial \rhos ^2} &\equiv& 0, \quad \gamma = 0 \quad or \quad \kappa_i = \kappa_j
\end{eqnarray}
\subsection{Proof of formulas (\ref{lem:dDdp_0}) and (\ref{lem:dDdp_1})}
Consider the partial derivative of $\Fmatrix$ with respect to the any correlation $\rhos$:
\begin{eqnarray}
\nonumber
\left( \frac{\partial \Fmatrix}{\partial \rhos}\right)' &=&
\frac{\partial }{\partial \rhos} \left(
2\Fmatrix  \Fmatrix - \delta \left(\kappamat \Fmatrix + \Fmatrix \Gammamat \right)
+ \frac{\delta(\delta - 1)}{2} \kappamat \Gammamat \right)
\\ \nonumber
&=&
2 \left( \frac{\partial \Fmatrix}{\partial \rhos} \Fmatrix + \Fmatrix \frac{\partial \Fmatrix}{\partial \rhos} \right)
- \delta \left(\kappamat \frac{\partial \Fmatrix}{\partial \rhos}
+ \frac{\partial \Fmatrix}{\partial \rhos} \Gammamat + \Fmatrix \frac{\partial \Gammamat}{\partial \rhos} \right)
\\ \nonumber
&+& \frac{\delta (\delta - 1)}{2} \kappamat \frac{\partial \Gammamat}{\partial \rhos}
\end{eqnarray}
Tending $\Corrmatrix$ to $\Imatrix$ we get:
\begin{eqnarray}
\nonumber
\lambdamat' &=& 2\left(\lambdamat \Psimat + \Psimat \lambdamat \right)
- \delta \left( \kappamat \lambdamat + \lambdamat \kappamat  + \Psimat \overbrace{ \left[ \kappamat \Is - \Is \kappamat\right]}^{under \,\,  Lemma \ref{lem:corr_properties},\,\ref{lem:corr_facts2}} \right)
\\ \nonumber
&+& \frac{\delta(\delta - 1)}{2} \kappamat
\underbrace{ \left[ \kappamat \Is - \Is \kappamat \right]}_{under \,\,  Lemma \ref{lem:corr_properties},\,\ref{lem:corr_facts2}}
\\ \nonumber
\lambdamat'_{ij} &=& 2\lambdamat_{ij}
\left(\Psimat_{ii} + \Psimat_{jj} -\delta \left[ \kappa_i + \kappa_j \right] \right)
-\delta
\sums \sumk \left( \Psimat_{is}\kappamat_{sk} \Is_{kj} - \Psimat_{is}\Is_{sk} \kappamat_{kj} \right)
\\ \nonumber
&+&
\frac{\delta(\delta - 1)}{2} \sums \sumk \left[ \kappamat_{is} \kappamat_{sk} \Is_{kj}
-\kappamat_{is} \Is_{sk} \kappamat_{kj}.
\right]
\\ \nonumber
\lambdamat'_{ij} &=& \lambdamat_{ij}
\left(2\Psimat_{ii} + 2\Psimat_{jj} -\delta \left[ \kappa_i + \kappa_j \right] \right)
-\delta \Psimat_{ii} \Is_{ij} \left[ \kappa_{i} -\kappa_{j} \right]
\\ \nonumber
&+&
\frac{\delta(\delta - 1)}{2}
\kappa_i \Is_{ij} \left[ \kappa_{i} - \kappa_{j}\right].
\\ \nonumber
\lambdamat'_{ij} &=& \lambdamat_{ij}
\left(2\Psimat_{ii} + 2\Psimat_{jj} -\delta \left[ \kappa_i + \kappa_j \right] \right)
-\delta \Is_{ij} \left[ \kappa_{i} -\kappa_{j} \right] \left[ \Psimat_{ii} + \frac{1 - \delta}{2}\kappa_i \right].
\\ \nonumber
\lambdamat_{ij}(0) &=& 0.
\end{eqnarray}
Since $\Is_{ij} = 0$ for $(ij) \notin \corridx$, hence $\lambdamat_{ij} \equiv 0$. Moreover, for diagonal elements $(ii) \notin \corridx$, $\forall i = 1..n$, therefore $\textbf{Tr}\lambdamat \equiv 0$.

\subsection{Proof of formula (\ref{lem:dDdp_2})}
\begin{eqnarray}
\left( \frac{\partial^2 \Fmatrix}{\partial \rhos \partial \rhok}\right)' &=&
\frac{\partial }{\partial \rhos \partial \rhok} \left(
2\Fmatrix  \Fmatrix - \delta \left(\kappamat \Fmatrix + \Fmatrix \Gammamat \right)
+ \frac{\delta(\delta - 1)}{2} \kappamat \Gammamat \right)
\\ \nonumber
&=&
2 \left( \frac{\partial^2 \Fmatrix}{\partial \rhos \partial \rhok} \Fmatrix+
\frac{\partial \Fmatrix}{\partial \rhos} \frac{\partial \Fmatrix}{\partial \rhok} +
\frac{\partial \Fmatrix}{\partial \rhok} \frac{\partial \Fmatrix}{\partial \rhos}
+ \Fmatrix \frac{\partial^2 \Fmatrix}{\partial \rhos \partial \rhok} \right)
\\ \nonumber
&-& \delta \Bigg(\kappamat \frac{\partial^2 \Fmatrix}{\partial \rhos \partial \rhok}
+ \frac{\partial^2 \Fmatrix}{\partial \rhos \partial \rhok} \Gammamat
+ \frac{\partial \Fmatrix}{\partial \rhos} \frac{\partial \Gammamat}{\partial \rhok}+
\frac{\partial \Fmatrix}{\partial \rhok} \frac{\partial \Gammamat}{\partial \rhos}
\\ \nonumber
&+& \Fmatrix \frac{\partial^2 \Gammamat}{\partial \rhos \partial \rhok}
\Bigg)
+ \frac{\delta (\delta - 1)}{2} \kappamat \frac{\partial^2 \Gammamat}{\partial \rhos \partial \rhok}
\end{eqnarray}
Let us define
\[
\etamat = \limrhoI \frac{\partial^2  \Fmatrix}{\partial \rhos \partial \rhok}, \quad
\tilde{\lambdamat} = \limrhoI \frac{\partial \Fmatrix}{\partial \rhok}
\]
therefore
\begin{eqnarray}
\nonumber
\etamat' &=& 2\left[ \etamat \Psimat + \Psimat \etamat \right]
- \delta \left[ \kappamat \etamat + \etamat \kappamat +
\lambdamat \left( \kappamat \Ik - \Ik \kappamat \right)
+ \tilde{\lambdamat} \left( \kappamat \Is - \Is \kappamat \right) +\Psimat \Qmatrix \right]
\\ \nonumber
&+&
\frac{\delta (\delta - 1)}{2} \kappamat \Qmatrix
\\ \nonumber
\etamat'_{ii} &=& 4 \etamat_{ii} \Psimat_{ii} - 2\delta \kappa_i \etamat_{ii} - \delta \Psimat_{ii} \Qmatrix_{ii}
\\ \nonumber
&-& \delta \sums\sumk
\left[
\lambdamat_{is} \kappamat_{sk} \Ik_{ki} -
\lambdamat_{is} \Ik_{sk} \kappamat_{ki} +
\tilde{\lambdamat}_{is} \kappamat_{sk} \Is_{ki} -
\tilde{\lambdamat}_{is} \Is_{sk} \kappamat_{ki}
\right]
+ \frac{\delta(\delta - 1)}{2} \kappa_i \Qmatrix_{ii}
\\ \nonumber
\etamat'_{ii} &=& 2 \etamat_{ii}\ \left[ 2 \Psimat_{ii}- \delta \kappamat_{ii} \right]
- \delta \sums \left[
\lambdamat_{is} \kappa_s \Ik_{si} -
\lambdamat_{is} \Ik_{si} \kappa_{i} +
\tilde{\lambdamat}_{is} \kappa_{s} \Is_{si} -
\tilde{\lambdamat}_{is} \Is_{si} \kappamat_{i}
\right]
\\ \nonumber
\etamat'_{ii} &=& 2 \etamat_{ii}\ \left[ 2 \Psimat_{ii}- \delta \kappamat_{ii} \right], \quad \etamat_{ii }(0) = 0.
\\ \nonumber
\etamat_{ii} &\equiv& 0.
\\ \nonumber
\textbf{Tr} \etamat &\equiv& 0.
\end{eqnarray}
\subsection{Proof of formulas (\ref{lem:dDdp_3})}
According to the definition of $\varphimat$ we obtain the following ODE:
\begin{eqnarray}
\nonumber
\varphimat' &=& 2\left[ \varphimat \Psimat + \Psimat \varphimat \right]
- \delta \left[ \kappamat \varphimat + \varphimat \kappamat + 2\lambdamat\left(\kappamat \Is - \Is \kappamat\right) +\Psimat \Pmatrix \right] +
\frac{\delta (\delta - 1)}{2} \kappamat \Pmatrix
\\ \nonumber
\varphimat(0) &=& \textbf{0}
\end{eqnarray}
or in the element wise notation:
\begin{eqnarray}
\nonumber
\varphimat'_{ii} &=& 2\varphimat_{ii} \left[2\Psimat_{ii} -\delta \kappamat_{ii}\right]  -2 \delta \lambdamat_{ij} \Is_{ij} (\kappa_j -\kappa_i) - \delta \Psimat_{ii} \Pmatrix_{ii}  + \frac{\delta (\delta -1)}{2} \kappa_i \Pmatrix_{ii}
\\ \nonumber
\varphimat'_{ii} &=& \varphimat_{ii} \left[4\Psimat_{ii} -2\delta \kappamat_{ii}\right] + 2 \delta \lambdamat_{ij} \Is_{ij} (\kappa_i -\kappa_j) - \delta \Pmatrix_{ii}\left( \Psimat_{ii}  + \frac{1- \delta}{2} \kappa_i \right)
\\ \nonumber
\varphimat'_{ii} &=& \varphimat_{ii} \left[ 4\Psimat_{ii} - 2\delta \kappamat_{ii}\right] + 2 \delta\Is_{ij} (\kappa_i -\kappa_j)
\left[\lambdamat_{ij} -
\kappa_i \frac{ (1-\sqrt{\delta})}{2}
\frac{e^{\expargi} + 1}{e^{\expargi} + \omega} \right]
\\ \nonumber
\varphimat(0) &=& 0.
\end{eqnarray}
It is easy to check that under the condition $\kappa_i = \kappa_j$:
\begin{equation}
\label{appD}
\varphimat_{ii} = \varphimat_{jj} = 0.
\end{equation}
The formula \ref{appD} also holds for the special case $\gamma = 0 (\delta = 1)$. Indeed, for this case $\lambdamat_{ij} = \lambdamat_{ji} = 0$. It turns out to that the RHS of the last equation for $\varphimat_{ij}$ is equal to zero, therefore $\varphimat_{ii} = \varphimat_{jj} = 0$.

We proceed with the case $i \notin \corridx$. Each element $\Pmatrix_{ii}$ equals $0$, i.e. $\varphimat_{ii}(\tau) \equiv 0$. Therefore, the trace of the matrix $\varphimat$ contains only two non-zero terms with multi-index $\corridx$. For simplicity of notation, we denote it as $i$ and $j$, i.e $\corridx = (ij)$. The summands $\varphimat_{ii}$ and $\varphimat_{jj}$ can be found via the following ODEs:
\begin{eqnarray}
\nonumber
\varphimat_{ii}' - \varphimat_{ii} \left[4\Psimat_{ii} -2\delta \kappa_i\right] &=&
2 \delta (\kappa_i -\kappa_j)
\left[\lambdamat_{ij} -
\kappa_i \frac{ (1-\sqrt{\delta})}{2}
\frac{e^{\expargi} + 1}{e^{\expargi} + \omega} \right].
\\ \nonumber
\varphimat_{jj}' - \varphimat_{jj} \left[4\Psimat_{jj} -2\delta \kappa_j\right] &=&
2 \delta (\kappa_j -\kappa_i)
\left[\lambdamat_{ji} -
\kappa_j \frac{ (1-\sqrt{\delta})}{2}
\frac{e^{\expargj} + 1}{e^{\expargj} + \omega} \right]
\\ \nonumber
\varphimat_{ii}(0) = \varphimat_{jj} (0) &=& 0.
\end{eqnarray}
Using Lemma \ref{lem:Phi_signs} we finish the proof.

\section{Auxiliary facts about the structure of the matrix $\Fmatrix$ in the zero correlation case}
Here we present some facts about the structure of $\Fmatrix$ for the zero correlation case. We consider the matrices $\Psimat$, $\lambdamat$ and $\varphimat$ defined as follows:
\begin{equation}
\label{eq:Psi_mat_def}
\Psimat = \limrhoI \Fmatrix
,\quad \lambdamat = \limrhoI \frac{\partial \Fmatrix}{\partial \rhos},
\quad \varphimat = \limrhoI \frac{\partial^2 \Fmatrix}{\partial \rhos^2}
\end{equation}
\begin{lemma}
	\label{lem:Psi_properties}
The matrix $\Psimat$ is a diagonal matrix with the following entries:
\begin{equation}
\nonumber
	\Psimat = diag
\left( \Psi(\kappa_1,\tau),
	   \Psi(\kappa_1,\tau),
	   \hdots,
	   \Psi(\kappa_n, \tau)\right)
\end{equation}
Here the function  $\Psi(\kappa, \tau)$ can be defined as a solution to the following one-dimensional Riccati equation
\begin{equation}
	\frac{d \Psi}{d\tau} = 2\Psi^2 -2\delta \kappa \Psi + \frac{\delta (\delta -1)} {2} \kappa^2, \quad \Psi(0) =0.
\end{equation}
which can be solved explicitly:
	\begin{equation}
	\label{eq:Psi_omega_def}
	\Psi(\kappa, \tau) = \frac{ \kappa \sqrt{\delta}(\sqrt{\delta} - 1)}{2} \frac{ e^{\exparg} - 1}{e^{\exparg} +\omega}, \quad \omega = \frac{1-\sqrt{\delta}}{1 + \sqrt{\delta}}.
	\end{equation}
Moreover,  the function $\Psi$ has the following properties:
	\begin{equation}
	\label{eq:int_Psi_def}
	\int \Psi(\kappa, \tau) d\tau = \frac{\delta + \sqrt{\delta}}{2} \kappa \tau - \frac{1}{2} \ln \left(e^{\exparg} + \omega \right) + C.
	\end{equation}
	\begin{equation}
	\label{eq:Psi_property}
	\Psi(\kappa, \tau) + \frac{1-\delta}{2}\kappa = \frac{\kappa(1-\sqrt{\delta)}}{2} \frac{e^{\exparg} + 1}{e^{\exparg} + \omega}
	\end{equation}
\end{lemma}
\begin{proof}
	\begin{eqnarray}
	\nonumber
	\frac{d\Psi}{d\tau} &=& 2 \Psi^2- 2 \delta \kappa \Psi +  \frac{\delta( \delta - 1)}{2} \kappa^2,
	\\ \nonumber
	d\tau &=&  \frac{d \Psi}{2 \Psi^2 - 2 \delta \kappa \Psi + \delta( \delta - 1) \kappa^2 /2 }
	\\ \nonumber
	\int d\tau &=&  \int \frac{d \Psi}{2 \Psi^2 - 2 \delta \kappa \Psi + \delta( \delta - 1) \kappa^2 / 2}
	\\ \nonumber
	\tau + c &=& \frac{1}{2\sqrt{\delta} \kappa}
	\left[
	\ln\left( \frac{\delta \kappa - 2\Psi}
	{\sqrt{\delta} \kappa} + 1\right)
	-
	\ln\left(1- \frac{\delta \kappa - 2\Psi}
	{\sqrt{\delta} \kappa} \right)
	\right]
	\\ \nonumber
	\tau + c &=& \frac{1}{2\sqrt{\delta} \kappa}
	\ln\left(
	\frac{\delta \kappa - 2\Psi + \sqrt{\delta} \kappa}
	{-\delta \kappa + 2\Psi + \sqrt{\delta} \kappa}\right)
	\\ \nonumber
	2\sqrt{\delta} \kappa \tau
	+ \ln\left(
	\frac{\delta \kappa + \sqrt{\delta} \kappa}
	{-\delta \kappa + \sqrt{\delta} \kappa}\right) &=&
	\ln\left(
	\frac{\delta \kappa - 2\Psi + \sqrt{\delta} \kappa}
	{-\delta \kappa + 2\Psi + \sqrt{\delta} \kappa}\right)
	\\ \nonumber
	2\sqrt{\delta} \kappa \tau &=&
	\ln\left(
	\frac{\delta \kappa - 2\Psi + \sqrt{\delta} \kappa}
	{-\delta \kappa + 2\Psi + \sqrt{\delta} \kappa}\right)
	- \ln\left(
	\frac{ 1 + \sqrt{\delta}}
	{1 - \sqrt{\delta}}\right)
	\\ \nonumber
	e^{2\sqrt{\delta} \kappa \tau} &=&
	\frac{\left(\delta \kappa - 2\Psi + \sqrt{\delta} \kappa\right)(1-\sqrt{\delta})}
	{\left( -\delta \kappa + 2\Psi + \sqrt{\delta} \kappa\right) (1 + \sqrt{\delta})}
	\\ \nonumber
	e^{2\sqrt{\delta} \kappa \tau} &=&
	\frac{2\Psi (\sqrt{\delta} - 1) + \sqrt{\delta}\kappa (1- \delta)}
	{2\Psi (\sqrt{\delta} + 1) + \sqrt{\delta}\kappa (1- \delta)}
	\end{eqnarray}
	Hence $\Psi$ equals to
	\begin{eqnarray}
	\nonumber
	\Psi  &=&\frac{1}{2} \frac{\sqrt{\delta}\kappa (1-\delta)\left(1 - e^{2\sqrt{\delta} \kappa \tau} \right)}
	{e^{2\sqrt{\delta} \kappa \tau}\left(1 + \sqrt{\delta} \right) + 1- \sqrt{\delta}}
	\\ 	\nonumber
	\Psi &=& -\frac{\sqrt{\delta} \kappa}{2} \frac{(1-\sqrt{\delta})\left(1 - e^{-2\sqrt{\delta} \kappa \tau} \right)}
	{1 + \frac{1- \sqrt{\delta}}{1+ \sqrt{\delta}} e^{-2\sqrt{\delta} \kappa \tau}}
	\\ 	\nonumber
	\Psi &=& \frac{\kappa \sqrt{\delta} (\sqrt{\delta} - 1)}{2} \frac{e^{\exparg} - 1} {e^{\exparg} + \omega}	
	, \quad \omega = \frac{1- \sqrt{\delta}}{1+ \sqrt{\delta}}
	\end{eqnarray}
\end{proof}

\begin{lemma}	
	Each element $\lambdamat_{ij}$ of the matrix $\lambdamat$ is the following function:
	\begin{eqnarray}
	 \nonumber
	\lambdamat_{ij} &=& \kappa_{i} \frac{\sqrt{\delta}  (1 - \sqrt{\delta})}{2(e^{\expargi} + \omega)(e^{\expargj} + \omega)} \times
	\\
	&\times& \Bigg[
	\frac{\kappa_j - \kappa_i}{\kappa_j + \kappa_i} \left(e^{(\kappa_j + \kappa_i) \sqrt{\delta} \tau} -1\right)\left(e^{(\kappa_j + \kappa_i) \sqrt{\delta} \tau} +\omega\right)
	\\ \nonumber
	&+&
	e^{\expargi}\left(e^{(\kappa_j- \kappa_i) \sqrt{\delta} \tau} -1\right)\left(e^{(\kappa_{j} - \kappa_i) \sqrt{\delta} \tau} +\omega\right)
	\Bigg]
	\end{eqnarray}
	\begin{proof}
	Differentiating the matrix equation \ref{eq:F_def} with respect to time $t$ and taking the limit $\Corrmatrix \rightarrow \Imatrix$, we get the following element wise ODEs for the $\lambdamat_{ij}$:
		\begin{eqnarray}
		\nonumber
		\lambdamat_{ij}' &=& \lambdamat_{ij}\left(2\Psimat_{ii} + 2\Psimat_{jj} -\delta \left[\kappa_{i} + \kappa_j \right] \right) -
		\delta \left[ \kappa_i - \kappa_j\right] \left[ \Psimat_{ii} + \frac{1-\delta}{2} \kappa_i\right]
		\\ \nonumber
		\lambdamat_{ij}(0) &=&0.
		\end{eqnarray}
The corresponding homogeneous ODE can be solved explicitly:
		\[
		\frac{e^{\kappa_i \sqrt{\delta} \tau + \kappa_j \sqrt{\delta} \tau}}{(e^{\expargi} + \omega)(e^{\expargj} + \omega)}.
		\]
		Thus, the solution of non-homogeneous problem reads
		\begin{eqnarray}
		\nonumber
		\lambdamat_{ij}
		&=&
		- \delta \left[\kappa_i - \kappa_j \right] \frac{\kappa_i (1-\sqrt{\delta})}{2} \frac{e^{\kappa_i \sqrt{\delta} \tau + \kappa_j \sqrt{\delta} \tau}}{(e^{\expargi} + \omega)(e^{\expargj} + \omega)}
		\\ \nonumber
		&\times& \int_0^\tau \frac{(e^{2\kappa_i \sqrt{\delta} \zeta} + 1)(e^{2\kappa_j \sqrt{\delta} \zeta} + \omega)}{e^{\kappa_i \sqrt{\delta} \zeta + \kappa_j \sqrt{\delta} \zeta}} d\zeta
		\\ \nonumber
		&=& - \delta \left[\kappa_i - \kappa_j \right] \frac{\kappa_i (1-\sqrt{\delta})}{2} \frac{e^{\kappa_i \sqrt{\delta} \tau + \kappa_j \sqrt{\delta} \tau}}{(e^{\expargi} + \omega)(e^{\expargj} + \omega)}
		\\ \nonumber
		&\times&
		\int_0^\tau \left[
		e^{(\kappa_i +\kappa_j)\sqrt{\delta} \zeta}
		+ \omega e^{(\kappa_i  - \kappa_j)\sqrt{\delta} \zeta}
		+e^{(\kappa_j - \kappa_i)\sqrt{\delta} \zeta}
		+ \omega e^{-(\kappa_i +\kappa_j)\sqrt{\delta} \zeta}
		\right]d\zeta;
		\\ \nonumber
		&=& \delta \left[\kappa_j - \kappa_i \right] \frac{\kappa_i (1-\sqrt{\delta})}{2} \frac{e^{(\kappa_i + \kappa_j) \sqrt{\delta} \tau }}{(e^{\expargi} + \omega)(e^{\expargj} + \omega)}
		\\ \nonumber
		&\times&
		\left[ \frac{e^{(\kappa_i +\kappa_j)\sqrt{\delta} \tau} - \omega e^{-(\kappa_i +\kappa_j)\sqrt{\delta} \tau} + \omega - 1}{(\kappa_i +\kappa_j)\sqrt{\delta} } +
		\frac{e^{(\kappa_j  - \kappa_i)\sqrt{\delta} \tau} - \omega e^{-(\kappa_j-\kappa_i)\sqrt{\delta} \tau} + \omega - 1 }{(\kappa_j -\kappa_i)\sqrt{\delta} }
		\right]
		\\ \nonumber
		&=& \kappa_{i} \frac{\sqrt{\delta}  (1 - \sqrt{\delta})}{2(e^{\expargi} + \omega)(e^{\expargj} + \omega)} \times
		\\ \nonumber
		&\times& \Bigg[
		\frac{\kappa_j - \kappa_i}{\kappa_j + \kappa_i} \left(e^{(\kappa_j + \kappa_i) \sqrt{\delta} \tau} -1\right)\left(e^{(\kappa_j + \kappa_i) \sqrt{\delta} \tau} +\omega\right)
		\\ \nonumber
		&+&
		e^{\expargi}\left(e^{(\kappa_j- \kappa_i) \sqrt{\delta} \tau} -1\right)\left(e^{(\kappa_{j} - \kappa_i) \sqrt{\delta} \tau} +\omega\right)
		\Bigg]
		\end{eqnarray}
	\end{proof}
\end{lemma}

\begin{lemma}
	\label{lem:Phi_signs}
	Any diagonal element $\varphimat_{ii}$ of the matrix $\varphimat$ can be defined as a solution to the following ODE:	
	\begin{eqnarray}
	\varphimat_{ii}' &=& - 2\kappa_i \sqrt{\delta} \frac{e^{\expargi} -\omega}{e^{\expargi}+\omega} \varphimat_{ii} +
	\delta (1 - \sqrt{\delta}) \kappa_i (\kappa_i -\kappa_j) \times
	\\ \nonumber
	&\times&
	\bigg[ -\frac{e^{\expargi} + 1}{e^{\expargi} + \omega} +
	\sqrt{\delta}\frac{\kappa_j - \kappa_i}{\kappa_j + \kappa_i} \
	\frac{
		e^{(\kappa_j + \kappa_i) \sqrt{\delta} \tau} -1 }{e^{\expargi} + \omega}
	\frac{e^{(\kappa_j + \kappa_i) \sqrt{\delta} \tau} +\omega}{e^{\expargj} + \omega}
	\\ \nonumber
	&+&\sqrt{\delta} e^{\expargi} \frac{e^{(\kappa_j- \kappa_i) \sqrt{\delta} \tau} -1 } {e^{\expargi} + \omega}
	\frac{ e^{(\kappa_{j} - \kappa_i) \sqrt{\delta} \tau} +\omega}
	{e^{\expargj} + \omega}
	\bigg]
	\\ \nonumber
	\varphimat_{ii}(0) &=&0
	\end{eqnarray}
	Moreover, the following inequalities holds 	for any $\kappa_i>0$, $\kappa_j >0$, $T>0$ and $\delta > 0$:
	\begin{eqnarray}
	\nonumber
	\int_0^{T} \left[{\varphimat_{ii}(u) + \varphimat_{jj}(u)} \right] du &>&0, \quad \delta >1, \quad \kappa_i \neq \kappa_j
	\\	\nonumber
	\int_0^{T} \left[{\varphimat_{ii}(u) + \varphimat_{jj}(u)} \right] du &\equiv& 0, \quad \delta =1 \quad or \quad \kappa_i = \kappa_j
	\\	\nonumber
	\int_0^{T} \left[{\varphimat_{ii}(u) + \varphimat_{jj}(u)} \right] du &<& 0, \quad 0<\delta <1 \quad \quad \kappa_i \neq \kappa_j
	\end{eqnarray}

\end{lemma}
\begin{proof}
	Can be checked by the direct calculations.
\end{proof}

\section{Auxilliary facts about correlation matrices}
\label{app:sec:corr_properties}
In this section we use two special types of square symmetric matrices, $\Is$ and $\Ixi$. These objects are defined as follows:
Matrix $\Is$ has zero entries, except elements with multiindex $(\corridx)$, these elements are equal to 1:
\begin{equation}
\label{eq:Is_def}
\Is_{ij} = 0, \forall (ij) \neq (\corridx), \quad \Is_{ij} = 1, (ij) = (\corridx), \quad or \quad (ji) = (\corridx).
\end{equation}
Matrix $\Is$ is a traceless matrix, $\textbf{Tr}\Is = 0$.
The matrix $\Ixi$ has also zero entries, except  only one element on $(u,u)$. This element is equal to 1.

We prove some useful facts about correlation matrix $\Corrmatrix$ and the similarity transform $\Gammamat = \Corrmatrix^{-1} \kappamat \Corrmatrix$ of the matrix $\kappamat$.
\begin{lemma}
	\label{lem:corr_properties}
	Correlation matrix $\Corrmatrix$ and its similarity transform $\Gammamat$ have the following properties:
	\begin{enumerate}
		\item{
			\label{lem:corr_facts1}
			\begin{equation}
			\frac{\partial \Corrmatrix^{-1}}{\partial \rhos} =
			-\Corrmatrix^{-1} \frac{\partial \Corrmatrix}{\partial \rhos} \Corrmatrix^{-1}.
			\end{equation}
		}
		\item{
			\label{lem:corr_facts2}
			\begin{equation}
			\limrhoI \frac{\partial \Gammamat}{\partial \rhos} = \kappamat \Is -\Is \kappamat.
			\end{equation}
		}
		\item{
			\label{lem:corr_facts3}
			\begin{equation}
			\limrhoI \frac{\partial^2 \Gammamat}{\partial \rhos \partial \rhok} = \Qmatrix, \quad
			\Qmatrix_{ii} = 0, \forall i = 1..n.
			\end{equation}
		}
		\item{
			\label{lem:corr_facts4}
			\begin{equation}
			\limrhoI \frac{\partial^2 \Gammamat}{\partial \rhos ^2} = \Pmatrix, \quad
			\Pmatrix_{ii} = 2 \Indicator(ij \in \corridx) \left[ \kappa_i - \kappa_j\right]
			\end{equation}
		}
	\end{enumerate}
\end{lemma}
\subsection{Proof of formula (\ref{lem:corr_facts1}).}
\begin{eqnarray}
\nonumber
\Corrmatrix \Corrmatrix^{-1} &=& \Imatrix
\\ \nonumber
\frac{\partial}{\partial \rhos} \left( \Corrmatrix \Corrmatrix^{-1} \right)
&=& \frac{\partial \Imatrix}{\partial \rhos}
\\ \nonumber
\Corrmatrix \frac{\partial \Corrmatrix^{-1} }{\partial \rhos}
&=& -  \frac{\partial \Corrmatrix}{\partial \rhos}   \Corrmatrix^{-1}
\\ \nonumber
\Corrmatrix \frac{\partial \Corrmatrix^{-1} }{\partial \rhos}
&=& -  \frac{\partial \Corrmatrix}{\partial \rhos}   \Corrmatrix^{-1}
\\ \nonumber
\frac{\partial \Corrmatrix^{-1} }{\partial \rhos}
&=& -  \Corrmatrix^{-1}\frac{\partial \Corrmatrix}{\partial \rhos}   \Corrmatrix^{-1}
\end{eqnarray}

\subsection{Proof of formula (\ref{lem:corr_facts2})}	
\begin{eqnarray}
\nonumber
\limrhoI \frac{\partial \Gammamat}{\partial \rhos} &=&
\limrhoI \frac{\partial \left(\Corrmatrix^{-1}\kappamat \Corrmatrix\right)}{\partial \rhos}
\\ \nonumber
&=& \limrhoI \frac{\partial \Corrmatrix^{-1}}{\partial \rhos} \kappamat \Corrmatrix +
\limrhoI \Corrmatrix^{-1} \kappamat \frac{\partial \Corrmatrix }{\partial \rhos}
\\ \nonumber
&=& \limrhoI \frac{\partial \Corrmatrix^{-1}}{\partial \rhos} \kappamat \Imatrix +
\Imatrix \kappamat \limrhoI \frac{\partial \Corrmatrix }{\partial \rhos}
\\ \nonumber
&=& -\limrhoI \frac{\partial \Corrmatrix}{\partial \rhos} \kappamat +
\kappamat \limrhoI \frac{\partial \Corrmatrix }{\partial \rhos}
\\ \nonumber
&=& -\Is \kappamat + \kappamat \Is
\\ \nonumber
&=& \kappamat \Is -\Is \kappamat.
\end{eqnarray}
\subsection{Proof of formula (\ref{lem:corr_facts3})}
\begin{eqnarray}
\frac{\partial^2 \Gammamat}{\partial \rhos \partial \rhok} &=&
\frac{\partial^2}{\partial \rhos \partial \rhok} \Corrmatrix^{-1} \kappamat \Corrmatrix
\\ \nonumber
&=&
\frac{\partial^2 \Corrmatrix^{-1}}{\partial \rhos \partial \rhok} \kappamat \Corrmatrix
+
\frac{\partial \Corrmatrix^{-1}}{\partial \rhos} \kappamat \frac{\partial \Corrmatrix}{\partial \rhok}
+
\frac{\partial \Corrmatrix^{-1}}{\partial \rhok} \kappamat \frac{\partial \Corrmatrix}{\partial \rhos}
+
\Corrmatrix^{-1} \kappamat \frac{\partial^2 \Corrmatrix}{\partial \rhos \partial \rhok}
\\ \nonumber
&=&
-\frac{\partial}{\partial \rhok}\left[\Corrmatrix^{-1} \frac{\partial \Corrmatrix}{\partial \rhos}  \Corrmatrix^{-1}\right] \kappamat \Corrmatrix
-
\Corrmatrix^{-1} \frac{\partial \Corrmatrix}{\partial \rhos} \Corrmatrix^{-1} \kappamat \frac{\partial \Corrmatrix}{\partial \rhok}
\\ \nonumber
&-&
\Corrmatrix^{-1} \frac{\partial \Corrmatrix}{\partial \rhok} \Corrmatrix^{-1} \kappamat \frac{\partial \Corrmatrix}{\partial \rhos}
\\ \nonumber
& = &
-\frac{\partial \Corrmatrix^{-1} }{\partial \rhok} \frac{\partial \Corrmatrix}{\partial \rhos}  \Corrmatrix^{-1} \kappamat \Corrmatrix
- \Corrmatrix^{-1} \frac{\partial \Corrmatrix}{\partial \rhos}  \frac{\partial \Corrmatrix^{-1} }{\partial \rhok} \kappamat \Corrmatrix
\\ \nonumber
&-&
\Corrmatrix^{-1} \frac{\partial \Corrmatrix}{\partial \rhos} \Corrmatrix^{-1} \kappamat \frac{\partial \Corrmatrix}{\partial \rhok}
-
\Corrmatrix^{-1} \frac{\partial \Corrmatrix}{\partial \rhok} \Corrmatrix^{-1} \kappamat \frac{\partial \Corrmatrix}{\partial \rhos}
\\ \nonumber
& = &
\Corrmatrix^{-1} \frac{\partial \Corrmatrix}{\partial \rhok}  \Corrmatrix^{-1} \frac{\partial \Corrmatrix}{\partial \rhos}  \Corrmatrix^{-1} \kappamat \Corrmatrix
+ \Corrmatrix^{-1} \frac{\partial \Corrmatrix}{\partial \rhos} \Corrmatrix^{-1} \frac{\partial \Corrmatrix}{\partial \rhok} \Corrmatrix^{-1} \kappamat \Corrmatrix
\\ \nonumber
&-&
\Corrmatrix^{-1} \frac{\partial \Corrmatrix}{\partial \rhos} \Corrmatrix^{-1} \kappamat \frac{\partial \Corrmatrix}{\partial \rhok}
-
\Corrmatrix^{-1} \frac{\partial \Corrmatrix}{\partial \rhok} \Corrmatrix^{-1} \kappamat \frac{\partial \Corrmatrix}{\partial \rhos}
\\ \nonumber
\Qmatrix &=&
\Ik \Is \kappamat  + \Is \Ik \kappamat - \Is \kappamat \Ik - \Ik \kappamat \Is
\\ \nonumber
\Qmatrix_{ii} &=& \sums \sumk \left[
\Ik_{is} \Is_{sk} \kappamat_{si}  +  \Is_{is} \Ik_{sk} \kappamat_{si} -
\Is_{is} \kappamat_{sk} \Ik_{ki} - \Ik_{is} \kappamat_{sk} \Is_{ki} \right].
\\ \nonumber
\Qmatrix_{ii} &=& \sums \left[
\Ik_{is} \Is_{si} \kappamat_{ii}  +  \Is_{is} \Ik_{si} \kappamat_{ii} -
\Is_{is} \kappamat_{ss} \Ik_{si} - \Ik_{is} \kappamat_{ss} \Is_{si} \right]
\\ \nonumber
\Qmatrix_{ii} &  = & 0.
\end{eqnarray}
Since $\Is_{is} = 0$ if $\Ik_{si} = 1$ for each $s = 1.. n$ and vice versa.

\subsection{Proof of formula (\ref{lem:corr_facts4})}
\begin{eqnarray}
\Pmatrix_{ii} &=& 2 \sums \left[ \Is_{is} \Is_{si} \kappamat_{ii} - \Is_{is} \kappamat_{ss} \Is_{si} \right]
\\ \nonumber
\Pmatrix_{ii} &=& 2 \Indicator(ij \in \corridx) \left[ \kappa_i - \kappa_j\right]
\end{eqnarray}

\end{document}